\RequirePackage{fix-cm}
\documentclass{svjour3}                     
\smartqed  
\usepackage{graphicx}
\usepackage{cite}
\usepackage{amsmath,amssymb,amsfonts}
\usepackage{algorithmic}
\usepackage{textcomp}
\usepackage{float}
\usepackage{subfigure}
\usepackage[colorlinks,linkcolor=blue,urlcolor=blue,citecolor=blue]{hyperref}
\begin{document}

\title{On Performance of Sparse Fast Fourier Transform Algorithms Using the Aliasing Filter}


\author{Bin Li$^{1}$, Zhikang Jiang$^{1}$ and Jie Chen$^{1}$}


\institute{Bin Li \at
              \email{sulibin@shu.edu.cn}           
           \and
           Zhikang Jiang \at
              \email{zkjiang@i.shu.edu.cn}  
           \and
           Jie Chen \at
              \email{jane.chen@shu.edu.cn}  
           \and
           $^{1}$ \quad School of Mechanical and Electrical Engineering and Automation, Shanghai University, Shanghai 200072, China
}

\date{Received: date / Accepted: date}

\maketitle

\begin{abstract}

Computing the Sparse Fast Fourier Transform(sFFT) of a $K$-sparse signal of size $N$ has emerged as a critical topic for a long time. There are mainly two stages in the sFFT: frequency bucketization and spectrum reconstruction. Frequency bucketization is equivalent to hashing the frequency coefficients into $B(\approx{K})$ buckets through one of these filters: Dirichlet kernel filter, flat filter, aliasing filter, etc. The spectrum reconstruction is equivalent to identifying frequencies that are isolated in their buckets. More than forty different sFFT algorithms compute Discrete Fourier Transform(DFT) by their unique methods so far. In order to use them properly, the urgent topic of great concern is how to analyze and evaluate the performance of these algorithms in theory and practice. The paper mainly discusses the sFFT Algorithms using the aliasing filter. In the first part, the paper introduces the technique of three frameworks: the one-shot framework based on the compressed sensing(CS) solver, the peeling framework based on the bipartite graph and the iterative framework based on the binary tree search. Then, we get the conclusion of the performance of six corresponding algorithms: sFFT-DT1.0, sFFT-DT2.0, sFFT-DT3.0, FFAST, R-FFAST and DSFFT algorithm in theory. In the second part, we make two categories of experiments for computing the signals of different SNR, different $N$, different $K$ by a standard testing platform and record the run time, percentage of the signal sampled and $L0$, $L1$, $L2$ error both in the exactly sparse case and general sparse case. The result of experiments satisfies the inferences obtained in theory.


\end{abstract}










\section{Introduction \label{Sect1}}
The widely popular algorithm to compute DFT is the fast Fourier transform(FFT) invented by Cooley and Tukey, which can compute a signal of size $N$ in $O(N$log$N$) time and use $O(N)$ samples. Nowadays, with the demand for big data computing and low sampling ratio, it motivates the require for the new algorithms to replace previous FFT algorithms that can compute DFT in sublinear time and only use a subset of the input data. The new sFFT algorithms take advantage of the signal's inherent characteristics that a large number of signals are sparse in the frequency domain, only $K(K<< N)$ frequencies are non-zeros or are significantly large. Under this assumption, people can reconstruct the spectrum with high accuracy by using only $K$ most significant frequencies. Because of its excellent performance and generally satisfied assumptions, the technology of sFFT was named one of the ten Breakthrough Technologies in MIT Technology Review in 2012.

Through frequency bucketization, $N$ frequency coefficients are hashed into $B$ buckets, and the length of one bucket is denoted by $L$. The effect of the Dirichlet kernel filter is to make the signal convoluted a rectangular window in the time domain; it can be equivalent to the signal multiply a Dirichlet kernel window of size $L(L<<N)$ in the frequency domain. The effect of the flat window filter is to make the signal multiply a mix window in the time domain; it can be equivalent to the signal convoluted a flat window of size $L(L<<N)$ in the frequency domain. The effect of the aliasing filter is to make the signal multiply a comb window in the time domain; it can be equivalent to the signal convoluted a comb window of size $B(B \approx{K})$ in the frequency domain. After bucketization, the algorithm only needs to focus on the non-empty buckets and locate the positions and estimate values of the large frequency coefficients in those buckets in what we call the identifying frequencies or the spectrum reconstruction. Up to now, there are more than forty sFFT algorithms using the sFFT idea and about ten sFFT algorithms using the aliasing filter. People are concerned about the performance of these algorithms including runtime complexity, sampling complexity and robustness performance. The results of these performance analyses are the guide for us to optimize these algorithms and use them selectively. The paper will give complete answers in theory and practice.

The first sFFT algorithms called the Ann Arbor fast Fourier transform(AAFFT) algorithm using the Dirichlet kernel filter is a randomized algorithm with runtime and sampling complexity $O(K^2 \text{poly(log}N))$. The performance of the AAFFT0.5 \cite{IEEEexample:Gilbert2002Near} algorithm was later improved to $O(K$poly(log$N))$ in the AAFFT0.9 \cite{IEEEexample:Iwen2007Empirical}, \cite{IEEEexample:Gilbert2008A} algorithm through the use of unequally-spaced FFTs and the use of binary search technique for spectrum reconstruction.

The sFFT algorithms using the flat window filter can compute the exactly $K$-sparse signal with runtime complexity $O(K\text{log}N)$ and general $K$-sparse signal with runtime complexity $O(K\text{log}N\text{log}(N/K))$. In the one-shot framework, the sFFT1.0 \cite{IEEEexample:Hassanieh2012} and sFFT2.0 \cite{IEEEexample:Hassanieh2012} algorithm can locate and estimate the $K$ largest coefficients in one shot. In the iterative framework, the sFFT3.0 \cite{IEEEexample:Hassanieh2012-2} and sFFT4.0 \cite{IEEEexample:Hassanieh2012-2} algorithm can locate the position by using only 2$B$ or more samples of the filtered signal inspired by the frequency offset estimation both in the exactly sparse case and general sparse case. Later, a new robust algorithm, so-called the Matrix Pencil FFT(MPFFT) algorithm, was proposed in \cite{IEEEexample:Chiu2014} based on the sFFT3.0 algorithm. The paper \cite{IEEEexample:Li2020} summarizes the two frameworks and five reconstruction methods of these five corresponding algorithms. Subsequently the performance of these five algorithms is analyzed and evaluated both in theory and practice.

There are three frameworks for sFFT algorithms using the flat window filter. The algorithms of the one-shot framework are the so-called sFFT by downsampling in the time domain(sFFT-DT)1.0 \cite{IEEEexample:Hsieh2013}, sFFT-DT2.0 \cite{IEEEexample:Hsieh2015}, sFFT-DT3.0 algorithm. The algorithms of the peeling framework are the so-called Fast Fourier Aliasing-based Sparse Transform(FFAST) \cite{IEEEexample:Pawar2013}, \cite{IEEEexample:Pawar2018} and R-FFAST(Robust FFAST) \cite{IEEEexample:Pawar2015}, \cite{IEEEexample:Ong2019} algorithm. The algorithm of the iterative framework is the so-called Deterministic Sparse FFT(DSFFT) \cite{IEEEexample:Plonka2018} algorithm. This paper mainly discusses the technology and performance of these six algorithms, and all the details will be mentioned later in the paper.

Under the assumption of arbitrary sampling (even the sampling interval is less than 1), the Gopher Fast Fourier Transform(GFFT) \cite{IEEEexample:Iwen2010}, \cite{IEEEexample:Iwen2013} algorithm and the Christlieb Lawlor Wang Sparse Fourier Transform(CLW-SFT) \cite{IEEEexample:LAWLOR2013}, \cite{IEEEexample:Christlieb2016} algorithm can compute the exactly $K$-sparse signal with runtime complexity $O(K\text{log}K)$. They are aliasing-based search deterministic algorithms guided by the Chinese Remainder Theorem(CRT). The DMSFT \cite{IEEEexample:Merhi2019}(generated from GFFT) algorithm and CLW-DSFT \cite{IEEEexample:Merhi2019}(generated from CLW-SFT) algorithm can compute the general $K$-sparse signal with runtime complexity $O(K^2\text{log}K)$. They use the multiscale error-correcting method to cope with high-level noise.

The paper \cite{IEEEexample:Gilbert2014} summarizes a three-step approach in the stage of spectrum reconstruction and provides a standard testing platform to evaluate different sFFT algorithms. There are also some researches try to conquer the sFFT problem from other aspects: complexity \cite{IEEEexample:Kapralov2017}, \cite{IEEEexample:Indyk2014}, performance \cite{IEEEexample:Lopez-Parrado2015}, \cite{IEEEexample:Chen2017}, software \cite{IEEEexample:Schumacher2014}, \cite{IEEEexample:Wang2016}, hardware \cite{IEEEexample:Abari2014}, higher dimensions \cite{IEEEexample:Wang2019}, \cite{IEEEexample:Kapralov2019}, implementation \cite{IEEEexample:Pang2018}, \cite{IEEEexample:Kumar2019} and special setting \cite{IEEEexample:Plonka2016}, \cite{IEEEexample:Plonka2017} perspectives.

The identification of different sFFT algorithms can be seen through a brief analysis as above. The algorithms using the Dirichlet kernel filter are not efficient because it only bins some frequency coefficients into one bucket one time. The algorithms using the flat filter are probabilistic algorithms with spectrum leakage. In comparison to them, the algorithms using the aliasing filter is very convenient and no spectrum leakage, whether $N$ is a product of some co-prime numbers or $N$ is a power of two. This type of algorithm is suitable for the exactly sparse case and general sparse case, but it is not easy to solve the worst case because there may be many frequency coefficients in the same bucket accidentally because the scaling operation is of no use for the filtered signal.

The paper is structured as follows. Section \ref{Sect2} provides a brief overview of the basic sFFT technique using the aliasing filter. Section \ref{Sect3} introduces and analyzes three frameworks and six corresponding spectrum reconstruction methods. In Section \ref{Sect4}, we analyze and evaluate the performance of six corresponding algorithms in theory. In the one-shot framework, sFFT-DT1.0, sFFT-DT2.0 and sFFT-DT3.0 algorithm uses the CS solver with the help of the Moment Preserving Problem(MPP) method. In the peeling framework, FFAST and R-FFAST algorithm uses the bipartite graph with the help of the packet erasure channel method. In the iterative framework, the DSFFT algorithm uses the binary tree search with the help of the statistical characteristics. In Section \ref{Sect5}, we do two categories of comparison experiments. The first kind of experiment is to compare them with each other. The second is to compare them with other algorithms. The analysis of the experiment results satisfies the inferences obtained in theory. 
\section{Preliminaries	\label{Sect2}}
In this section, we initially present some notations and basic definitions of sFFT. 
\subsection{Notation}
The $N$-th root of unify is denoted by $\omega _{N}=e^{-2\pi \mathbf{i}/N}$. The DFT matrix of size $N$ is denoted by $\mathbf{F}_{N}\in \mathbb{C}^{N\times N}$ as follows:
\begin{equation}
\mathbf{F}_{N}[j,k]=\frac{1}{N}\omega _{N}^{jk}	\label{Eq1}
\end{equation}

The DFT of a vector $x\in \mathbb{C}^{N}$ (consider a signal of size $N$) is a vector $\hat{x}\in \mathbb{C}^{N}$ defined as follows:
\begin{equation}
\begin{split}
\hat{x}=\mathbf{F}_{N}x	\\
\hat{x}_{i}=\frac{1}{N} \sum_{j=0}^{N-1}x_{j}\omega _{N}^{ij}	
\end{split}	\label{Eq2}
\end{equation}

In the exactly sparse case, spectrum $\hat x$ is exactly $K$-sparse if it has exactly $K$ non-zero frequency coefficients while the remaining $N-K$ coefficients are zero. In the general sparse case, spectrum $\hat x$ is general $K$-sparse if it has $K$ significant frequency coefficients while the remaining $N-K$ coefficients are approximately equal to zero. The goal of sFFT is to recover a $K$-sparse approximation $\hat {x'}$ by locating $K$ frequency positions $f_0, \dots, f_{K-1}$ and estimating $K$ largest frequency coefficients $\hat{x}_{f_{0}},\dots,\hat{x}_{f_{K-1}}$.

\subsection{Techniques of frequency bucketization using the aliasing filter}
The first stage of sFFT is encoding by frequency bucketization. The process of frequency bucketization using the aliasing filter is achieved through shift operation and subsampling operation.

The first technique is the use of shift operation. The offset parameter is denoted by $\tau \in \mathbb{R}$. The shift operation representing the original signal multiplied by matrix $\mathbf{S}_{\tau}$. $\mathbf{S}_{\tau}\in \mathbb{R}^{N\times N}$ is defined as follows:
\begin{equation}
\mathbf{S}_{\tau}[j,k]=\left\{\begin{matrix}  1, & j-\tau\equiv k(\text{mod} N) 
\\0,&\text{o.w.}
\end{matrix}\right.	\label{Eq3}
\end{equation}

If a vector $x'\in \mathbb{C}^{N}$, $x'=\mathbf{S}_{\tau} x$, $\hat{x'}=\mathbf{F}_{N}\mathbf{S}_{\tau} x$, such that:
\begin{equation}
\begin{split}
x'_i=x_{(i-\tau)}	\\
x'_{i+\tau}=x_i	\\
\hat{x'}_{i}=\hat{x}_{i}\omega ^{ \tau i}
\end{split}	\label{Eq4}
\end{equation}

The second technique is the use of the aliasing filter. The signal in the time domain is subsampled such that the corresponding spectrum in the frequency domain is aliased. It also means frequency bucketization. The subsampling factor is denoted by $L\in\mathbb{Z}^{+}$, representing how many frequencies aliasing in one bucket. The subsampling number is denoted by $B\in\mathbb{Z}^{+}$, representing the number of buckets$(B= N/L)$. The subsampling operation representing the original signal multiplied by matrix $\mathbf{D}_{L}$. $\mathbf{D}_{L}\in \mathbb{R}^{B\times N}$ is defined as follows:
\begin{equation}
\mathbf{D}_{L}[j,k]=\left\{\begin{matrix}  1, & k=jL
\\0,&\text{o.w.}
\end{matrix}\right.	\label{Eq5}
\end{equation}	

Let vector $\hat{y}_{B,\tau}\in\mathbb{C}^{B}$ be the filtered spectrum obtained by shift operation and subsampling operation. If $\hat{y}_{B,\tau}=\mathbf{F}_{B}\mathbf{D}_{L}\mathbf{S}_{\tau}x$, we get Equation (\ref{Eq6}) in bucket $i$.
\begin{equation}
\hat y _{B,\tau }[i]=\sum_{j=0}^{L-1}{\hat {x} _{jB+i} \omega ^{\tau(jB+i)}}	\label{Eq6}
\end{equation}	

As we see above, frequency bucketization includes three steps: shift operation($x'=\mathbf{S}_{\tau} x$, it costs 0 runtime), subsampling operation($y_{B,\tau}=\mathbf{D}_{L}\mathbf{S}_{\tau}x$, it costs $B$ samples),  Fourier Transform($\hat{y}_{B,\tau}=\mathbf{F}_{B}\mathbf{D}_{L}\mathbf{S}_{\tau}x$, it costs $B\text{log}B$ runtime). So totally frequency bucketization one round costs $B\text{log}B$ runtime and $B$ samples.

\section{Techniques	\label{Sect3}}
In this section, we introduce an overview of the techniques and frameworks that we will use in the sFFT algorithms based on the aliasing filter.

As mentions above, frequency bucketization can decrease runtime and sampling complexity because all operations are calculated in $B$ dimensions$(B=O(K), B<<N)$. After frequency bucketization, it needs spectrum reconstruction by identifying frequencies that are isolated in their buckets. The aliasing filter may lead to frequency aliasing where more than one significant frequencies are aliasing in one bucket. It increases the difficulty in recovery because finding frequency position and estimating vales of frequency become indistinguishable in terms of their aliasing characteristic. There are three frameworks to overcome this problem, the one-shot framework, the peeling framework, the iterative framework.  

\subsection{One-shot framework based on the CS solver}
Firstly we introduce the one-shot framework which can recover all $K$ significant frequencies in one shot. The block diagram of the one-shot framework is shown in Figure \ref{fig1}.
\begin{figure}[H]
\centering
\includegraphics[width=15 cm]{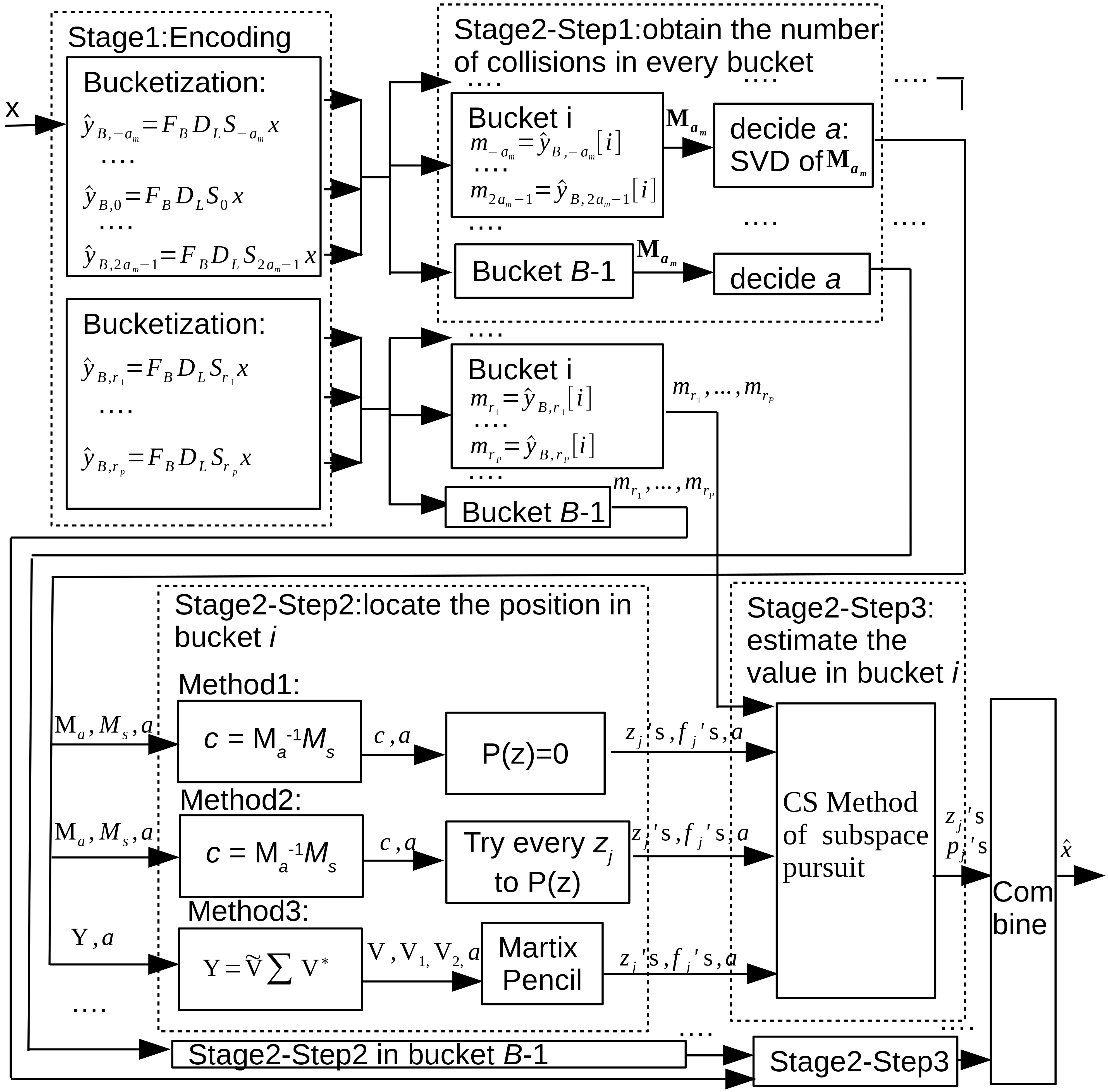}
\caption{A system block diagram of the one-shot framework.\label{fig1}}
\end{figure}  

The first stage of sFFT is encoding by frequency bucketization. Suppose there are at most $a_m$ number of significant frequencies aliasing in every bucket, after running 3$a_m$ times' round for set $\tau =\{ \tau_0=-a_m,\tau_1=-a_m+1,\dots,\tau_{a_m}=0,\dots,\tau_{3a_{m}-1}=2a_{m}-1 \}$, calculate $\hat{y}_{B,\tau}=\mathbf{F}_{B}\mathbf{D}_{L}\mathbf{S}_{\tau}x$ representing filtered spectrum by encoding. Suppose in bucket $i$, the number of significant frequencies is denoted by $a$, it is a high probability that $a \leq a_m$, we get simplified Equation (\ref{Eq9}) from Equation (\ref{Eq7}) and Equation (\ref{Eq8}). In Equation (\ref{Eq8}) and Equation (\ref{Eq9}), $p_j=\hat{x}_{f_{j}}$ respecting effective frequency values for $ |p_0| \geq |p_1| \geq \dots \geq |p_{a-1}| \gg \text{other values of }p_j$, $z_j=\omega^{f_j}$ respecting effective frequency position for $f_j \in \left \{ i,i+L,\dots,i+(L-1)B) \right \}$, $m_k=\hat y _{B,k }[i]$ respecting filtered spectrum in bucket $i$ for $k \in \left \{ -a_m,\dots,2a_m-1 \right \}$. In most cases, $a=0$ respecting sparsity. In a small number of cases, $a=1$ respecting only one significant frequency in the bucket. Only in very little cases, $a >=2$ respecting frequencies aliasing in the bucket. It's very unlikely that $a>a_m$. In the exactly sparse case, the approximately equal sign becomes the equal sign in Equation (\ref{Eq7}), Equation (\ref{Eq8}), Equation (\ref{Eq9}), Equation (\ref{Eq11}), Equation (\ref{Eq14}) and Equation (\ref{Eq18}).
\begin{equation}
\hat y _{B,\tau }[i]=\sum_{j=0}^{L-1}{\hat {x} _{jB+i} \omega ^{\tau(jB+i)}}
\approx \sum_{j=0}^{a_{m}-1}{\hat {x}_{f_j}}\omega ^{ \tau f_j}	\label{Eq7}
\end{equation}	
\begin{equation}
m_k \approx \sum_{j=0}^{a_{m}-1}p_j z_j^{k} \approx \sum_{j=0}^{a-1}p_j z_j^{k}		\label{Eq8}
\end{equation}	
\begin{equation}
\begin{bmatrix}
z_0^{0}& z_1^{0} & \cdots &z_{a-1}^{0} \\ 
z_0^{1}& z_1^{1} & \cdots &z_{a-1}^{1} \\ 
\cdots & \cdots & \cdots & \cdots\\ 
z_0^{2a-1}& z_1^{2a-1} & \cdots &z_{a-1}^{2a-1}
\end{bmatrix}
\begin{bmatrix}
p_0\\ 
p_1\\
\cdots\\
p_{a-1} 
\end{bmatrix}
\approx 
\begin{bmatrix}
m_0\\ 
m_1\\
\cdots\\
m_{2a-1} 
\end{bmatrix}	\label{Eq9}
\end{equation}	

The spectrum reconstruction problem in bucket $i$ is equivalent to obtain unknown variables including $a,z_0,\cdots,z_{a-1},p_0,\cdots,p_{a-1}$ for we have known variables including $m_{-a},\cdots,m_{2a-1}$. The aliasing problem is reformulated as Moment Preserving Problem(MPP). The MPP problem formulated by Bose-Chaudhuri-Hocquenghem(BCH) codes can be divided into three subproblems: how to obtain $a$, how to obtain $z_j$'s, how to obtain $p_j$'s in every bucket. Later, we will solve these three subproblems step by step.

Step 1: Obtain the number of significant frequencies in bucket $i$.

Solution: Suppose $a \leq  a_m$, it means $m_k$ is composed of at most $a_m$ number of Prony component $p_j$'s. Let vector $\hat z_j\in \mathbb{C}^{ a_m  \times 1}$ defined as $\hat{z}_j=[z_j^{0},z_j^{1},\dots,z_j^{a_m-1}]^{T}$ and Matrix $\mathbf{M}_{a_{m}}\in \mathbb{C}^{a_m \times a_m}$ defined as Equation (\ref{Eq10}), then the relationship between $\mathbf{M}_{a_m}$, $\hat z_j$ and $\hat z_j^T$ satisfies Theorem 1.
\begin{equation}
\mathbf{M}_{a_m}=\begin{bmatrix}
m_0 &m_1  &\cdots  & m_{a_m -1}\\ 
m_1 &m_2  &\cdots  & m_{a_m }\\ 
\cdots & \cdots & \cdots & \cdots\\ 
m_{a_m -1} &m_{a_m}  &\cdots  & m_{2a_m -1}
\end{bmatrix}_{{a_m}\times {a_m}}	\label{Eq10}
\end{equation}
\begin{lemma}
\begin{equation}
\mathbf{M}_{a_{m}}\approx \sum_{j=0}^{a_m-1}{p_j \hat{z}_j \hat{z}_j^{T}}	\label{Eq11}
\end{equation}
\end{lemma}
\begin{proof}[Proof of Theorem 1]
\begin{equation*}
\begin{split} 
&p_0\hat{z}_0z_0^{0}+p_1\hat{z}_1z_1^{0}+\dots+p_{a_{m}-1}\hat{z}_{a_{m}-1}z_{a_{m}-1}^{0}\approx[m_0,m_1,\dots,m_{a_{m}-1}]^{T} \\
&p_0\hat{z}_0z_0^{1}+p_1\hat{z}_1z_1^{1}+\dots+p_{a_{m}-1}\hat{z}_{a_{m}-1}z_{a_{m}-1}^{1}\approx[m_1,m_2,\dots,m_{a_{m}}]^{T}		\\
& \dots	\\
&p_0\hat{z}_0z_0^{a_{m}-1}+p_1\hat{z}_1z_1^{a_{m}-1}+\dots+p_{a_{m}-1}\hat{z}_{a_{m}-1}z_{a_{m}-1}^{a_{m}-1}\approx[m_{a_{m}-1},m_{a_{m}},\dots,m_{2a_{m}}]^{T}		\\
\end{split} 
\end{equation*}
  Based on the properties as mentioned above, we get Equation (\ref{Eq11}).
\end{proof}
Equation (\ref{Eq11}) is similar to the symmetric singular value decomposition(SSVD). Nevertheless, there are some differences. 1) $p_j$'s are complex but not real. 2) The columns of $[\hat{z}_0,\dots,\hat{z}_{a_m-1} ]$ are not mutually orthogonal normalized vectors. It is easy to do a transformation that $\mathbf{M}_{a_{m}}\approx \sum_{j=0}^{a_m-1}{p'_j \hat{z '}_j \hat{z '}_j^{T}}$, where $p_j'$'s are real and the absolute value of $p_j'$ is directly proportional to the $\left | p_j \right |$, and the columns of   $[\hat{z'}_0,\dots,\hat{z'}_{a_m-1} ]$ are mutually orthogonal normalized vectors. The paper \cite{IEEEexample:Bunse-Gerstner1988} proved that for the symmetric matrix, the  $\Sigma$ got from the SVD is equal to the $\Sigma$ got from the SSVD. For example, the SVD of $\begin{bmatrix} 0 &1 \\  1& 0 \end{bmatrix}$ is $\begin{bmatrix} 0 &1 \\  1& 0 \end{bmatrix}= \begin{bmatrix} 0 &1 \\ 
 1& 0 \end{bmatrix} \begin{bmatrix} 1 &0 \\  0& 1 \end{bmatrix} \begin{bmatrix} 1 &0 \\  0& 1\end{bmatrix}$ and the SSVD of $\begin{bmatrix} 0 &1 \\  1& 0 \end{bmatrix}$ is $\begin{bmatrix} 0 &1 \\  1& 0
\end{bmatrix}= \frac{1}{\sqrt{2}} \begin{bmatrix}1 &-i \\  1& i\end{bmatrix}\begin{bmatrix}1 &0 \\  0& 1\end{bmatrix}\frac{1}{\sqrt{2}}\begin{bmatrix}1 &1 \\  -i& i\end{bmatrix}$; the $\Sigma$ gained by these two ways are the same. After knowing it, we can compute the SVD of $\mathbf{M}_{a_{m}}$ and obtain $a_m$ singular values, then do the principal component analysis(PCA), $\sigma_1 \geq  \sigma_2 \geq  \dots \geq  \sigma_{a_m}$. In these $a_m$ number of singular values $\sigma_j$'s, how many large singular values mean how many efficient Prony components $p_j$'s; it also indicates how many significant frequencies in bucket $i$.

Step 2: Obtain effective frequency position $f_j$'s in bucket $i$.
    
Solution: Let the orthogonal polynomial formula $P(z)$ defined as Equation (\ref{Eq12}) and $P(z) \approx 0$. Let Matrix $\mathbf{M}_{a}\in \mathbb{C}^{a \times a}$ defined as Equation (\ref{Eq13}). Let vector $C$ defined as $C=[c_0,c_1,\dots,c_{a-1}]^T$. Let vector $M_s$ defined as $M_s=[-m_a,-m_{a+1},\dots$,
$-m_{2a-1}]^T$. The moments' formula satisfies Theorem 2.
\begin{equation}
P(z)= z^a+c_{a-1}z^{a-1}+ \dots + c_1z+c_0	\label{Eq12}
\end{equation}
\begin{equation}
\mathbf{M}_{a}=\begin{bmatrix}
m_0 &m_1  &\cdots  & m_{a -1}\\ 
m_1 &m_2  &\cdots  & m_{a }\\ 
\cdots & \cdots & \cdots & \cdots\\ 
m_{a -1} &m_{a_m}  &\cdots  & m_{2a -2}
\end{bmatrix}_{{a}\times {a}}			\label{Eq13}
\end{equation}
\begin{lemma}	$\mathbf{{M}_{a}}C \approx M_s$
\begin{equation}
\begin{bmatrix}
m_0 &m_1  &\cdots  & m_{a -1}\\ 
m_1 &m_2  &\cdots  & m_{a }\\ 
\cdots & \cdots & \cdots & \cdots\\ 
m_{a -1} &m_{a_m}  &\cdots  & m_{2a -2}
\end{bmatrix}_{{a}\times {a}}\begin{bmatrix}
c_0\\ 
c_1\\ 
\dots\\ 
c_{a-1}
\end{bmatrix}_{a \times 1}\approx \begin{bmatrix}
-m_a\\ 
-m_{a+1}\\ 
\dots \\ 
-m_{2a-1}
\end{bmatrix}_{a \times 1}		\label{Eq14}
\end{equation}
\end{lemma}
\begin{proof}[Proof of Theorem 2]
\begin{equation*}
\begin{split} 
&c_0m_0 \approx  (p_{0}z_{0}^{0}+ \dots p_{a-1}z_{a-1}^{0})c_0 \\
&\dots		\\
& c_{a-1}m_{a-1} \approx (p_{0}z_{0}^{a-1}+ \dots p_{a-1}z_{a-1}^{a-1})c_{a-1}	\\
&\Rightarrow c_0m_0+ \dots + c_{a-1}m_{a-1}		\\
& \approx  p_0(c_{0}z_{0}^{0}+ \dots c_{a-1}z_{0}^{a-1})+\dots+p_{a-1}(c_{0}z_{a-1}^{0}+ \dots c_{a-1}z_{a-1}^{a-1})\\
& \approx  (-p_0z_0^{a})+\dots+(-p_{a-1}z_{a-1}^{a}) \approx  -m_a\\
\end{split} 
\end{equation*}
The first element of $M_s$ has been proved and other elements of $M_s$ can also be proved.
\end{proof}
For the convenience of matrix calculation, for a matrix, the superscript "T" denotes the transpose, the superscript "+" denotes Moore-Penrose inverse or pseudoinverse, the superscript "*" denotes adjoint matrix, the superscript "-1" denotes inverse. Through Theorem 2, we can obtain $C \approx (\mathbf{M}_a^{-1}) M_s$. After gaining $C$, there are three ways to obtain $z_j$'s through Equation (\ref{Eq12}). The first approach is the polynomial method, the second approach is the enumeration method, and the last approach is the Matrix Pencil method. After knowing $z_j$'s, we can obtain approximate positions $f_j$'s through $z_j=\omega^{f_j}$.

Method 1) Polynomial method: In the exactly sparse case, the $a$ number of roots of $P(z)=0$ is the solution of $z_j$'s . For example if $a=1$, through $P(z) = z+c_0= 0$, then $z_0 = -c_0$. If $a = 2$, through $P(z) = z^2+ c_1z+c_0 = 0$, then $z_0=(-c_1-(c_1^2-4c_0)^{0.5})/2, z_1=(-c_1+(c_1^2-4c_0)^{0.5})/2$.

Method 2) Enumeration method: For $f_j \in \left \{ i,i+L,\dots,i+(L-1)B) \right \}$, try every possible position $z_j=\omega^{f_j}$ to the Equation (\ref{Eq12}), the first $a$ number of smallest $|P(z)|$ is the solution of $z_j$'s. $L$ times' attempts are needed in one solution.

Method 3) Matrix Pencil method: The method is proposed in paper \cite{IEEEexample:Hua1990}. For our problem, let the Toeplitz matrix $\mathbf{Y}$ defined as Equation (\ref{Eq15}), let $\mathbf{Y_1}$ be $\mathbf{Y}$ with the rightmost column removed defined as Equation (\ref{Eq16}), let $\mathbf{Y_2}$ be $\mathbf{Y}$ with the leftmost column removed defined as Equation (\ref{Eq17}). The set of generalized eigenvalues of $\mathbf{Y_2}-\lambda\mathbf{Y_1}$ satisfies Theorem 3.
\begin{equation}
\mathbf{Y}=\begin{bmatrix}
m_0 & m_{-1} &\dots  & m_{-a}\\ 
m_1 & m_0 & \dots & m_{-a+1}\\ 
\dots &\dots  & \dots & \dots\\ 
m_a &m_{a-1}  & \dots & m_0
\end{bmatrix}_{(a+1)\times(a+1)}		\label{Eq15}
\end{equation}
\begin{equation}
\mathbf{Y}_1=\begin{bmatrix}
m_0 & m_{-1} &\dots  & m_{-a+1}\\ 
m_1 & m_0 & \dots & m_{-a+2}\\ 
\dots &\dots  & \dots & \dots\\ 
m_a &m_{a-1}  & \dots & m_1
\end{bmatrix}_{(a+1)\times(a)}		\label{Eq16}
\end{equation}
\begin{equation}
\mathbf{Y}_2=\begin{bmatrix}
m_{-1} & m_{-2} &\dots  & m_{-a}\\ 
m_0 & m_{-1} & \dots & m_{-a+1}\\ 
\dots &\dots  & \dots & \dots\\ 
m_{a-1} &m_{a-2}  & \dots & m_0
\end{bmatrix}_{(a+1)\times(a)}		\label{Eq17}
\end{equation}
\begin{lemma} The set of generalized eigenvalues of $\mathbf{Y_2}-\lambda\mathbf{Y_1}$ are the $z_j$'s we seek.
\end{lemma}
\begin{proof}[Proof of Theorem 3]
Let the diagonal matrix $\mathbf{C}\in \mathbb{C}^{a \times a}$ defined as  $\mathbf{C}=\text{diag}(p_j)$, Let the Vandermonde Martix $\mathbf{U}_{a+1}\in \mathbb{C}^{(a+1) \times a}$ defined as follows:
$\mathbf{U}_{a+1}=\begin{bmatrix}
 1& 1 & \dots  &1 \\ 
z_0 &z_1  & \dots & z_{a-1}\\ 
 \dots &  \dots &  \dots &  \dots\\ 
 z_0^a& z_1^a & \dots  & z_{a-1}^a
\end{bmatrix}_{(a+1) \times a}$. $\mathbf{Y},\mathbf{Y_1},\mathbf{Y_2}$ has a Vandermonde decomposition, we can get $\mathbf{Y}=\frac{1}{a+1}\mathbf{U}_{a+1}\mathbf{C}\mathbf{U}_{a+1}^*$, $\mathbf{Y}_1=\frac{1}{a+1}\mathbf{U}_{a+1}\mathbf{C}\mathbf{U}_{a}^*$, $\mathbf{Y}_2=\frac{1}{a+1}\mathbf{U}_{a+1}\mathbf{C}( \text{diag}(z_j)^*)\mathbf{U}_{a}^*$. For example, if $a$=1, $\mathbf{Y}=\begin{bmatrix}
m_0 &m_{-1} \\ m_1 &m_0 \end{bmatrix}=\begin{bmatrix}p_0 & p_0 z_0^{-1}\\ 
 p_0 z_0^{1}& p_0\end{bmatrix}=\begin{bmatrix}1\\ z_0\end{bmatrix}
p_0\begin{bmatrix}1 & z_0^{-1}\end{bmatrix}$, and if $a$=2, $\mathbf{Y}=\begin{bmatrix}m_0 & m_{-1}  & m_{-2}\\ 
m_1 & m_0 & m_{-1}\\ m_2 & m_1 & m_0 \end{bmatrix}=
\begin{bmatrix} p_0+p_1 &  p_0z_0^{-1}+p_1z_1^{-1}& p_0z_0^{-2}+p_1z_1^{-2} \\ 
p_0z_0^{1}+p_1z_1^{1} & p_0+p_1 &p_0z_0^{-1}+p_1z_1^{-1} \\ 
p_0z_0^{2}+p_1z_1^{2} & p_0z_0^{1}+p_1z_1^{1} & p_0+p_1 \end{bmatrix}=
\begin{bmatrix} 1 & 1\\ z_0 &z_1 \\ z_0^2 & z_1^2 \end{bmatrix}
\begin{bmatrix} p_0 & \\ &p_1 \end{bmatrix}
\begin{bmatrix} 1 & z_0^{-1} & z_0^{-2} \\ 1& z_1^{-1}  & z_1^{-2} 
\end{bmatrix} $, $\mathbf{Y}_1=\begin{bmatrix}
m_0 & m_{-1}  \\m_1 & m_0 \\ m_2 & m_1 \end{bmatrix}=\begin{bmatrix} 1 & 1\\ 
z_0 &z_1 \\ z_0^2 & z_1^2 \end{bmatrix} \begin{bmatrix} p_0 & \\  &p_1 
\end{bmatrix} \begin{bmatrix} 1 & z_0^{-1} \\ 1& z_1^{-1}  
\end{bmatrix}$, $\mathbf{Y}_2=\begin{bmatrix} m_{-1} & m_{-2}  \\ 
m_0 & m_{-1} \\ m_1 & m_0 \end{bmatrix}=
\begin{bmatrix} 1 & 1\\ z_0 &z_1 \\ 
z_0^2 & z_1^2 \end{bmatrix}
\begin{bmatrix} p_0 & \\ &p_1 \end{bmatrix}
\begin{bmatrix} z_0^{-1} & \\ &z_1^{-1}\end{bmatrix}
\begin{bmatrix} 1 & z_0^{-1} \\ 1& z_1^{-1}  
\end{bmatrix}$. Using the Vandermonde decomposition, we can get $\mathbf{Y}_2-\lambda \mathbf{Y}_1=\frac{1}{a+1}\mathbf{U}_{a+1}\mathbf{C} (\text{diag}(z_j)^*-\lambda \mathbf{I})\mathbf{U}_{a}^*$, so the Theorem 3 can be proved.
\end{proof}

If the rank$(\mathbf{Y})=a$, the set of generalized eigenvalues of $\mathbf{Y_2}-\lambda\mathbf{Y_1}$ is equal to the set of nonzero$^2$ ordinary eigenvalues of $(\mathbf{Y_1}^+)\mathbf{Y_2}$. It is most likely that the rank$(\mathbf{Y})<a$, it is necessary to compute the SVD of the $\mathbf{Y}$, $\mathbf Y = \widetilde{\mathbf V} \Sigma \mathbf V ^*$, then we can use the Matrix Pencil Method to deal with the Matrix $\mathbf{V}$ afterword. For details, please refer to paper \cite{IEEEexample:Hua1990}, \cite{IEEEexample:Chiu2014}.

Step 3: Obtain effective frequency values $p_j$'s in bucket $i$.

Solution: In order to use the CS method, we need several random sampling. So for $P$ number of random numbers $r_1,\dots,r_P$, we calculate  $\hat{y}_{B,\tau}=\mathbf{F}_{B}\mathbf{D}_{L}\mathbf{S}_{\tau}x$ for set $\tau =\left\{ \tau_0=r_1,\tau_1=r_2,\dots, \tau_{P-1}=r_{P}          \right\}$ in $P$ times' round. Suppose in bucket $i$, the number of significant frequencies $a$ and approximate effective frequency position $z_j$'s have been known by step 1 and step 2; we can get Equation (\ref{Eq18}). (There are maybe errors for $z_j$'s obtained by step 2 because of interference of other $L-a$ number of Prony components.)  
\begin{equation}
\begin{bmatrix}
m_1\\ \dots\\ m_P\end{bmatrix}_{P\times 1}=\begin{bmatrix}
z_{0}^{r_{1}} &\dots  & z_{L-1}^{r_1}\\  \dots& \dots & \dots\\ 
z_{0}^{r_{P}} & \dots & z_{L-1}^{r_{P}}\end{bmatrix}_{P\times L} 
\begin{bmatrix}p_0\\ \dots \\ p_{L-1}\end{bmatrix}_{L\times 1}\approx 
\begin{bmatrix}z_{0}^{r_{1}} &\dots  & z_{a-1}^{r_1}\\ 
 \dots& \dots & \dots\\ z_{0}^{r_{P}} & \dots & z_{a-1}^{r_{P}}
\end{bmatrix}_{P\times a} \begin{bmatrix}p_0\\ \dots \\ p_{a-1}
\end{bmatrix}_{a\times 1}		\label{Eq18}
\end{equation}

The Equation (\ref{Eq18}) is very likely similar to the CS formula. The model of CS is formulated as $y\approx \mathbf{\Phi} S$, where $S$ is a sparse signal, $\mathbf{\Phi}$ is a sensing matrix, $y$ is the measurements. In Equation (\ref{Eq18}), $y$ is a vector of $P \times 1$ by $P$ measurements, $\mathbf{\Phi}$ is a matrix of $P \times L$, $S$ is a vector of $L \times 1$ which is $a$-sparse. It should be noted that $\mathbf{\Phi}$ must satisfy either the restricted isometry property (RIP) or mutual incoherence property (MIP) for the successful recovery with high probability. It has been known that the Gaussian random matrix and partial Fourier matrix are good candidates to be $\mathbf{\Phi}$, so the $\mathbf{\Phi}$ of the Equation (\ref{Eq18}) meets the criteria. Furthermore, the number of measurements $P$ one collect should more than $a\text{log}_{10}L$, so that these measurements will be sufficient to recover signal $x$. 

In order to obtain $p_j$'s by CS solver, we use the subspace pursuit method. The process is as follows: 1) Through the positions $f_j$'s gained by step 2, get the possible value of $z_j$'s as follows ${z}'_j=\omega^{f_j}$, ${z}''_j=\omega^{f_j+B}$, ${z}'''_j=\omega^{f_j-B}$, then obtain 3$a$ vectors as follows: $\{ {{z}'_0}^{r_1} ,\dots,{{z}'_0}^{r_P} \}^T,\{ {{z}''_0}^{r_1} ,\dots,{{z}''_0}^{r_P} \}^T$,
$\{ {{z}'''_0}^{r_1} ,\dots,{{z}'''_0}^{r_P} \}^T,\dots\{ {{z}'''_{a-1}}^{r_1} ,\dots,{{z}'''_{a-1}}^{r_P} \}^T$. An (over-complete) dictionary can be characterized by a matrix $\mathbf{D}$, and it contains $3a$ vectors listed above. (one wishes one-third vectors of them form a basis). Each (column) vector in a dictionary is called an atom. 2) From $3a$ atoms of the dictionary matrix $\mathbf{D}$, find $a$ number of atoms that best match the measurements' residual error. Select these $a$ number of atoms to construct a new sensing matrix $\hat{\mathbf{\Phi}}$. 3) Obtain $\hat S$ by the support of $\hat S= \text{argmin}\left \| y-\hat{\Phi } \hat S \right \|_2$ through the least square method(LSM). 4) If the residual error $\left \| y-\hat{\Phi } \hat S \right \|_2$ meets the requirement, or the number of iterations reaches the assumption, or the residual error becomes larger, the iteration will be quit, otherwise continue to step 2. After computing, we get the final sensing matrix $\hat{\mathbf{\Phi}}$ of size $P \times a$ and sparse signal $\hat S$ of size $a$ just in the right part of Equation (\ref{Eq18}). So we get effective frequency positions $z_j$'s and effective frequency values $p_j$'s in bucket $i$.

\subsection{Peeling framework based on the bipartite graph}
Secondly, we introduce the peeling framework which can recover all $K$ significant frequencies layer by layer. The block diagram of the peeling framework is shown in Figure \ref{fig2}.
\begin{figure}[H]
\centering
\includegraphics[width=15 cm]{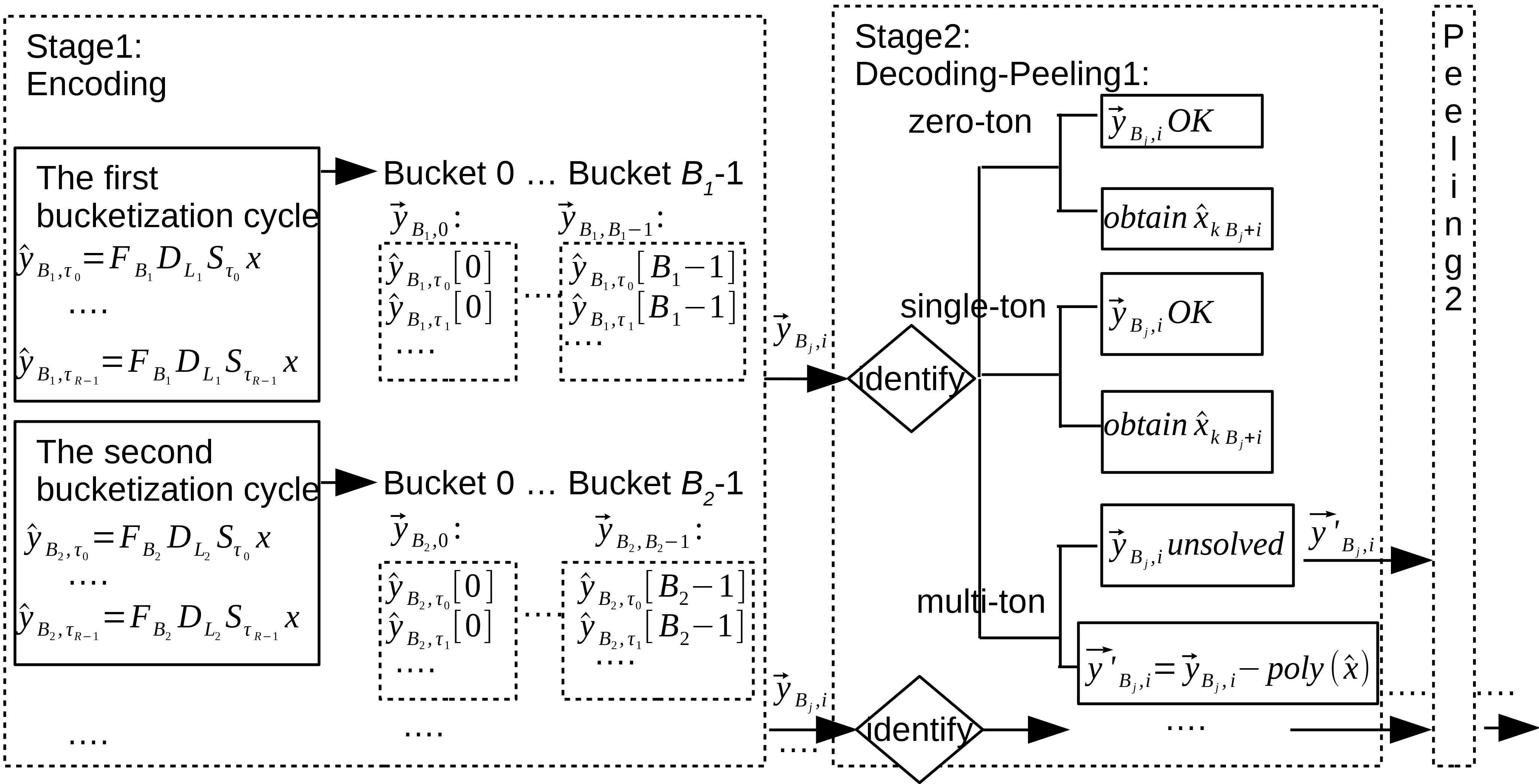}
\caption{A system block diagram of the peeling framework.\label{fig2}}
\end{figure}  
In the peeling framework, we require the size $N$ of signal $x$ to be a product of a few(typically three or more) co-prime numbers $p_i$'s. For example $N=p_1p_2p_3$ where $p_1, p_2, p_3$ are co-prime numbers. With this precondition, we determine $L_i$'s and $B_i$'s gained from $p_i$'s in each bucketization cycle. In every cycle, we use the same set $\tau =\left\{ \tau_0=r_1,\tau_1=r_2,\dots, \tau_{R-1}=r_{R}          \right\}$ inspired by different spectrum reconstruction ways introduced later to calculate $\hat{y}_{B,\tau}=\mathbf{F}_{B}\mathbf{D}_{L}\mathbf{S}_{\tau}x$ in $R$ times' rounds(it also means $R$ delay chains for one bucket). Suppose there are $d$ times' cycles, after $d$ cycles, stage 1 encoding by bucketization is completed. In order to solve the aliasing problems, the filtered spectrum in bucket $i$ of the No.$j$' cycle is denoted by vector $\overrightarrow{y}_{B_j,i}\in \mathbb{C}^{R \times 1}$ as follows:
\begin{equation}
\overrightarrow{y}_{B_j,i}=\begin{bmatrix}
\hat{y}_{B_j,\tau_0}[i]\\ 
\dots\\ 
\hat{y}_{B_j,\tau_{R-1}}[i]
\end{bmatrix}=\begin{bmatrix}
\sum_{k=0}^{L_j-1}\hat x _{(kB_j+i)} \omega^{\tau_0(kB_j+i)}\\ 
\dots\\ 
\sum_{k=0}^{L_j-1}\hat x _{(kB_j+i)} \omega^{\tau_{R-1}(kB_j+i)}
\end{bmatrix}		\label{Eq19}
\end{equation}

We use a simple example to illustrate the process of encoding by bucketization. Consider a twenty($N$ = 20) size signal $x$ that has only five($K$ = 5) significant coefficients $\hat x _1,\hat x _3,\hat x _5,\hat x _{10},\hat x _{13}>>0$, while the rest of the coefficients are approximately equal to zero. With this precondition, there are two bucketization cycles. In the first cycle, for $B_1=4$ and $L_1=5$, we get four vectors $\{\overrightarrow{y}_{4,0},\overrightarrow{y}_{4,1},\overrightarrow{y}_{4,2},\overrightarrow{y}_{4,3}\}$ respecting the filtered spectrum in four buckets for set $\tau $ in $R$ times' rounds. In the second cycle, for $B_2=5$ and $L_2=4$, we get five vectors $\{\overrightarrow{y}_{5,0},\overrightarrow{y}_{5,1},\overrightarrow{y}_{5,2},\overrightarrow{y}_{5,3},\overrightarrow{y}_{5,4}\}$. After the bucketization, we can construct a bipartite graph shown in Figure \ref{fig3} through Equation (\ref{Eq19}). In Figure \ref{fig3}, there are $N=20$ variable nodes on the left(referred to twenty coefficients of $\hat{x}$) and $N_b=B_1+B_2=4+5=9$ parity check nodes on the right(referred to nine buckets in two cycles). The values of the parity check nodes on the right are approximately equal to the complex sum of the values of variable nodes(its left neighbors) through Equation (\ref{Eq19}). In these check nodes, some have no significant variable node, as no left neighbor, is called a 'zero-ton' bucket(three blue-colored check nodes). Some have exactly one significant variable node, as one left neighbor, is called a 'single-ton' bucket(three green-colored check nodes). Others have more than one significant variable nodes, as more than one left neighbors, is called a 'multi-ton' bucket(three red-colored check nodes).
\begin{figure}[H]
\centering
\includegraphics[width=5 cm]{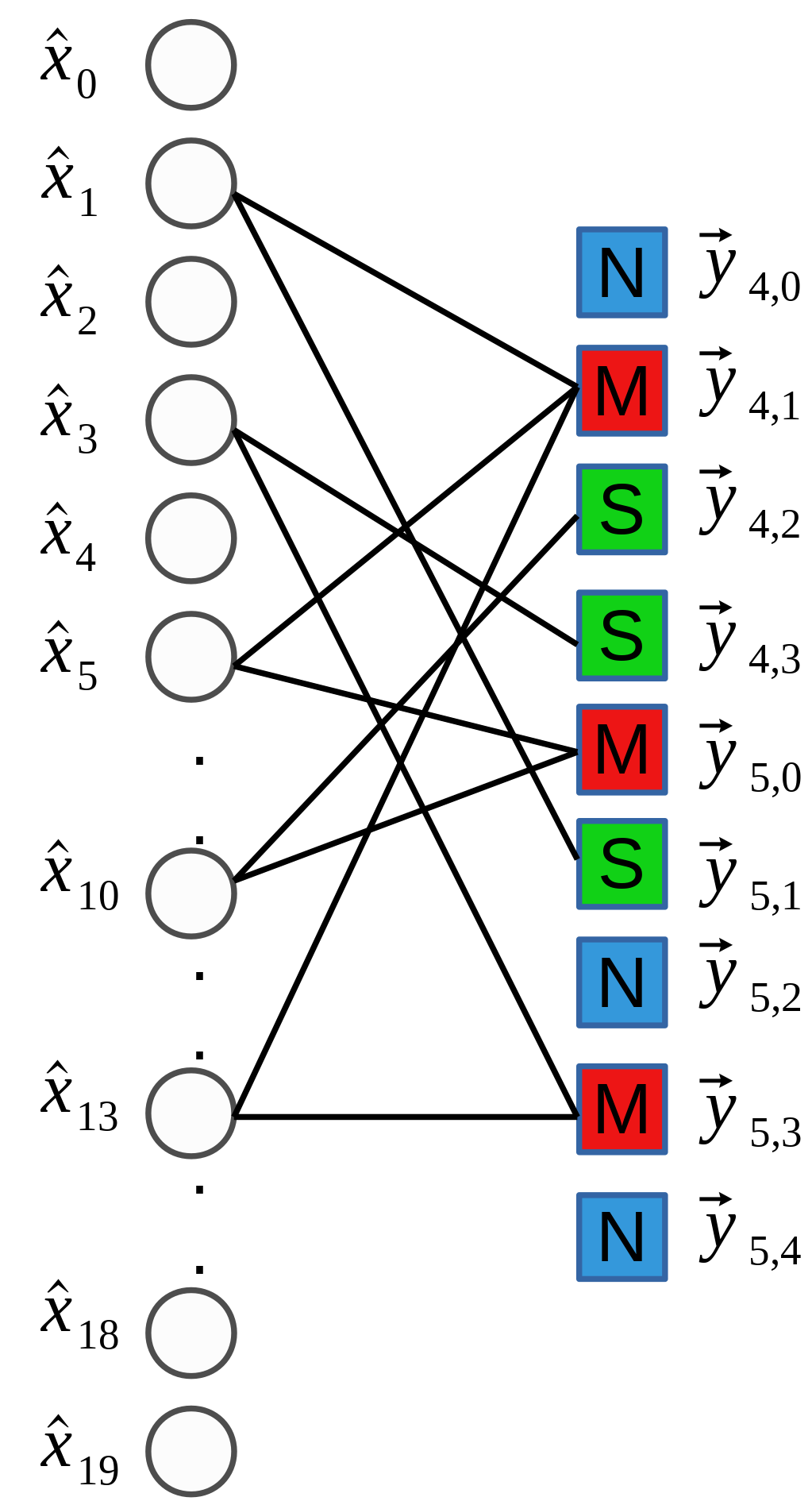}
\caption{An example of bipartite graph with $N=20$ variable nodes on the left and $N_b=9$ parity check nodes on the right(blue square with "N" means 'zero-ton' bucket, green square with "S" means 'single-ton' bucket, red square with "M" means 'multi-ton' bucket).\label{fig3}}
\end{figure}  

After bucketization, the subsequent problem is how to recover the spectrum from all buckets gained from several cycles. Through the identification of vector $\overrightarrow{y}_{B_j,i}$, we can determine the characteristics of bucket $B_j[i]$. If the bucket is a 'zero-ton' bucket, then $\hat x_{(kB_j+i)}=0   \  \text{for} \  k \in[0,L_j-1] $, so the problem of frequency recovery in this bucket can be solved. If the bucket is a 'single-ton' bucket, suppose the one effective frequency position $f_0$ and the one effective frequency value $\hat{x}_{f_{0}}$ can be obtained by the afterward methods, then $\hat x_{(kB_j+i)}=\left\{\begin{matrix}
0   \ \ \text{for} \ \ (k \in[0,L_j-1]) \cap (kB_j+i  \neq f_0)\\ 
\hat x_{f_0}   \ \ \text{for} \ \ kB_j+i =f_0
\end{matrix}\right.   $, so the frequency recovery in this bucket can be solved. If the bucket is a 'multi-ton' bucket, it is necessary to separate the 'multi-ton' bucket into the 'single-ton' bucket to realize the decoding of the 'multi-ton' bucket. For example, after the first peeling, in the bucket $i$ of the No.$j$' cycle, suppose $\hat{x}_{f_0},\dots \hat{x}_{f_{n_1-1}}$ have been got by other 'single-ton' bucket, then the vector $\overrightarrow{y}_{B_j,i}$ respecting original bucket changed to $\overrightarrow{y}'_{B_j,i}$,  where $\overrightarrow{y}'_{B_j,i}=\overrightarrow{y}_{B_j,i}-\text{poly}(\hat{x}_{f_0})-\dots -\text{poly}(\hat{x}_{f_{n_{1}-1}})$ respecting remaining frequencies in the bucket. Through the identification of vector $\overrightarrow{y}'_{B_j,i}$, we can analyze the remaining elements in this bucket; if it is a 'zero-ton' bucket or a 'single-ton' bucket, we can stop peeling and get all frequencies in this bucket. If not, continue the second peeling, suppose $\hat{x}_{f_{n_1}},\dots \hat{x}_{f_{n_2-1}}$ can be got by other 'single-ton' bucket through new peeling. We can identify $\overrightarrow{y}''_{B_j,i}$ to continue to analyze the bucket, where $\overrightarrow{y}''_{B_j,i}=\overrightarrow{y}'_{B_j,i}-\text{poly}(\hat{x}_{f_{n_1}})-\dots -\text{poly}(\hat{x}_{f_{n_{2}-1}})$. After $q$ times' peeling, the problem of frequency recovery in the 'multi-ton' bucket can be solved as follows:
\begin{equation}
\hat x_{(kB_j+i)}=\left\{\begin{matrix}
0   \ \ \text{for} \ \ (k \in[0,L_j-1]) \cap (kB_j+i  \neq f_0) \cap \dots \cap (kB_j+i  \neq f_{n_q-1})\\ 
\hat x_{f_0}   \ \ \text{for} \ \ kB_j+i =f_0 \\
\dots
\\
\hat x_{f_{n_q-1}}   \ \ \text{for} \ \ kB_j+i =f_{n_q-1}
\end{matrix}\right.  \label{Eq20}
\end{equation}

If the frequency recovery in all buckets can be solved, we can finish the spectrum reconstruction. The peeling-decoder successfully recovers all the frequencies with high probability under these three following assumptions: 1) 'Zero-ton'; 'single-ton'; and 'multi-ton' buckets can be identified correctly.  2) If the bucket is a 'single-ton' bucket, the decoder can locate the effective frequency position $f_0$ and estimate the effective frequency value $\hat{x}_{f_{0}}$ correctly. 3) It is sufficient to cope with 'multi-ton' buckets by the peeling platform.

Subproblem 1: How to identify vector $\overrightarrow{y}_{B_j,i}$ to distinguish bucket $B_j[i]$?

Solution: In the exactly sparse case, if $ \left \| \overrightarrow{y}_{B_j,i} \right \|_2=0$, the bucket  is a 'zero-ton' bucket. If the bucket  is not a 'zero-ton' bucket, make the second judgment. The way to make the second judgment about whether the bucket is a 'single-ton' bucket or not is to judge  $\left \|\overrightarrow{y}_{B_j,i}-\text{poly}(\hat{x}_{f_0})\right \|_2=0 \ \text {or} \ \neq 0$, where $\hat x_{f_0}$ is gained from the solution of subproblem 2. If the bucket $B_j[i]$ is not a 'zero-ton' bucket either a 'single-ton' bucket, it is a 'multi-ton' bucket. In the general sparse case, the two equations are translated to $ \left \| \overrightarrow{y}_{B_j,i} \right \|_2\leq T_0$ and $\left \|\overrightarrow{y}_{B_j,i}-\text{poly}(\hat{x}_{f_0}) \right \|_2\leq  T_1$, where $T_0$ and $T_1$ are the identification thresholds. It costs $O(R)$ runtime to identify vector $\overrightarrow{y}_{B_j,i}$ by knowing $\hat x_{f_0}$ in advance. After $d$ times' cycles, it costs $O(R(B_1+\dots+B_d))$ runtime for the identification in the first peeling.

Subproblem 2: Suppose the target is a 'single-ton' bucket, how to recover the one frequency in this bucket?

Solution: In the exactly sparse case, we use $R = 3$ and the set $\tau =\left\{ \tau_0=0,\tau_1=1,\tau_2=2          \right\}$ to calculate $\hat{y}_{B,\tau}=\mathbf{F}_{B}\mathbf{D}_{L}\mathbf{S}_{\tau}x$ in three rounds. Suppose the bucket $B_j[i]$ is a 'single-ton' bucket, three elements of the vector $\overrightarrow{y}_{B_j,i}$ are $\hat{y}_{B_j,0}[i]=\hat x_{f_0} \omega^{0\cdot f_0}$, $\hat{y}_{B_j,1}[i]=\hat x_{f_0} \omega^{1\cdot f_0}$ and $\hat{y}_{B_j,2}[i]=\hat x_{f_0} \omega^{2\cdot f_0}$. So only if $ \frac{\hat{y}_{B_j,0}[i]}{\hat{y}_{B_j,1}[i]}= \frac{\hat{y}_{B_j,1}[i]}{\hat{y}_{B_j,2}[i]}$, it is a single-ton bucket. If it is a single-ton bucket, the position $f_0$ can be obtained by $ f_0=\angle \frac{\hat{y}_{B_j,1}[i]}{\hat{y}_{B_j,0}[i]}\cdot (\frac{-N}{2\pi})$, and the value $\hat x_{f_0}$ can be obtained by $\hat x_{f_0} = \hat{y}_{B_j,0}[i]$. In all, it costs three samples and four runtime to recover the one significant frequency in one bucket. 

In the general sparse case, the frequency recovery method is the optimal whitening filter coefficients of the Minimum Mean Squared Error(MMSE) estimator and sinusoidal structured bin-measurement matrix for the speedy recovery. At first, the set of the one candidate $f_0$ is $f_0 \in \left \{ i,i+L,\dots,i+(L-1)B) \right \}$, and there are $L$ possible choices. In the first iteration of binary-search, we choose a random number $r_0$, then calculate $\hat{y}_{B_j,r_0},\hat{y}_{B_j,r_0+1},\dots,\hat{y}_{B_j,r_0+m-1}$ in $m$ times' rounds. In fact, we obtain $\hat{y}_{B_j,r_0}[i] \approx  \hat x_{f_0} \omega^{(r_0)\cdot f_0},\hat{y}_{B_j,r_0+1}[i] \approx  \hat x_{f_0} \omega^{(r_0+1)\cdot f_0},\dots,\hat{y}_{B_j,r_0+m-1}[i] \approx  \hat x_{f_0} \omega^{(r_0+m-1)\cdot f_0}$. $\Delta _t $ respecting the phase difference is defined as $\Delta _t = \angle \hat{y}_{B_j,r_0+t+1}[i] - \angle \hat{y}_{B_j,r_0+t}[i] = \angle \omega ^ {f_0}+U_{t+1}-U_{t}$ where $U_t$ respecting real error in No.$t$' round. In paper \cite{IEEEexample:Kay1989}, we can see the Maximum Likelihood Estimate(MLE) $\angle \hat \omega ^ {f_0}$ of $\angle \omega ^ {f_0}$ is calculated by Equation (\ref{Eq21}), where $w_t$ is defined as Equation (\ref{Eq22}). After obtaining $\angle \hat \omega ^ {f_0}$, make a judgment of binary-search; if $\angle \hat \omega ^ {f_0} \in[0,\pi]$, there is a new restriction that $f_0 \in [0,N/2-1]$, otherwise the restriction is $f_0 \in [N/2,N-1]$ the complement set of set $[0,N/2-1]$. The next iteration is very similar, we choose a random number $r_1$, then calculate $\hat{y}_{B_j,r_1}[i],\hat{y}_{B_j,r_1+2}[i],\dots,\hat{y}_{B_j,r_1+2(m-1)}[i]$ in bucket $i$. After obtaining $\angle \hat \omega ^ {2f_0}$ through Equation (\ref{Eq21}) and Equation (\ref{Eq22}), make a judgment of binary-search; if $\angle \hat \omega ^ {2f_0} \in[0,\pi]$, there is a new restriction that $f_0 \in [0,N/4-1]\cup [N/2,3N/4-1]$, otherwise the restriction is the complement set of the previous set. After $C$ number of iterations, in the advantage of restrictions, we can locate the only one position $f_0$ from the original set $ \left \{ i,i+B,\dots,i+(L-1)B) \right \}$. From the paper \cite{IEEEexample:Pawar2015}, if the number of iterations $C = O(\text{log} N)$ and the number of rounds per iteration $m = O(\text{log}^{1/3}N)$, the singleton-estimator algorithm can correctly identify the unknown frequency $f_0$ with a high probability. If the signal noise ratio(SNR) is low, the $m$ and $C$ must be increased. After knowing $f_0$, we can obtain $\hat x_{f_0} $ by applying the LSM. In all, to recover the one approximately significant frequency in one bucket, it needs $Cm = O(\text{log}^{4/3}N)$ samples and $3Cm = O(\text{log}^{4/3}N)$ runtime. The runtime includes $Cm$ runtime to calculate all $\Delta _t $'s and $2Cm$ runtime to calculate all $\angle \hat \omega ^ {2^j{f_0}}$'s. $\hat x_{f_0}$ can be used to judge whether the bucket is a 'single-ton' bucket or not through the discriminant $\overrightarrow{y}_{B_j,i}-\text{poly}(\hat{x}_{f_0})\leq T_1$.
\begin{equation}
\angle \hat \omega ^ {f_0}=\sum_{t=0}^{m-2}w_t\Delta _t \label{Eq21}
\end{equation}
\begin{equation}
w_t=\frac{3m}{2(m^2-1)}(1-4(\frac{2t-m+2}{2m})^2) \label{Eq22}
\end{equation}

Subproblem 3: How to solve the 'multi-ton' buckets by the peeling platform?

Solution with method 1: The genie-assisted-peeling decoder by the bipartite graph is useful. As shown in Figure \ref{fig3}, the bipartite graph represents the characteristics of the bucketization. The variable nodes on the left represent the characteristics of the frequencies. The parity check nodes on the right represent the characteristics of the buckets. Every efficient variable node connects $d$ different parity check nodes as its neighbor in $d$ time's cycles; it respects the $d$ edges connected to each efficient variable node. The example of the process of the genie-assisted-peeling decoder is shown in Figure \ref{fig4}, and the steps are as follows:  

Step 1: We can identify all the indistinguishable buckets. If the bucket is a 'zero-ton' bucket, the frequency recovery in this bucket is finished. If the bucket is a 'single-ton' bucket, we can obtain the frequency in the bucket through the solution of subproblem 2. In the graph, these operations represent to select and remove all right nodes with degree 0 or degree 1, and moreover, to remove the edges connected to these right nodes and corresponding left nodes.

Step 2: We should remove the contributions of these frequencies gained by step 1 in other buckets. For example, the first peeling in bucket $i$ means we calculate a new vector $\overrightarrow{y'}_{B_j,i}$ just as $\overrightarrow{y'}_{B_j,i}=\overrightarrow{y}_{B_j,i}-\text{poly}(\hat{x}_{f_0})-\dots -\text{poly}(\hat{x}_{f_{n_{1}-1}})$, where $\hat{x}_{f_0},\dots \hat{x}_{f_{n_1-1}}$ have been gained by step 1. In the graph, these operations represent to remove all the other edges connected to the left nodes removed in step 1. When the edges are removed, their contributions are subtracted from their right nodes. In the new graph, the degree of some right nodes decreases as well.

Step 3: If all buckets have been identified, we successfully reconstruct the spectrum $\hat x$. Otherwise, we turn to step 1 and step 2 to continue identifying and peeling operations. In the graph, it means if all right nodes are removed, the decoding is finished. Otherwise, turn to step 1 and step 2.
\begin{figure}[H]
\centering
\includegraphics[width=14 cm]{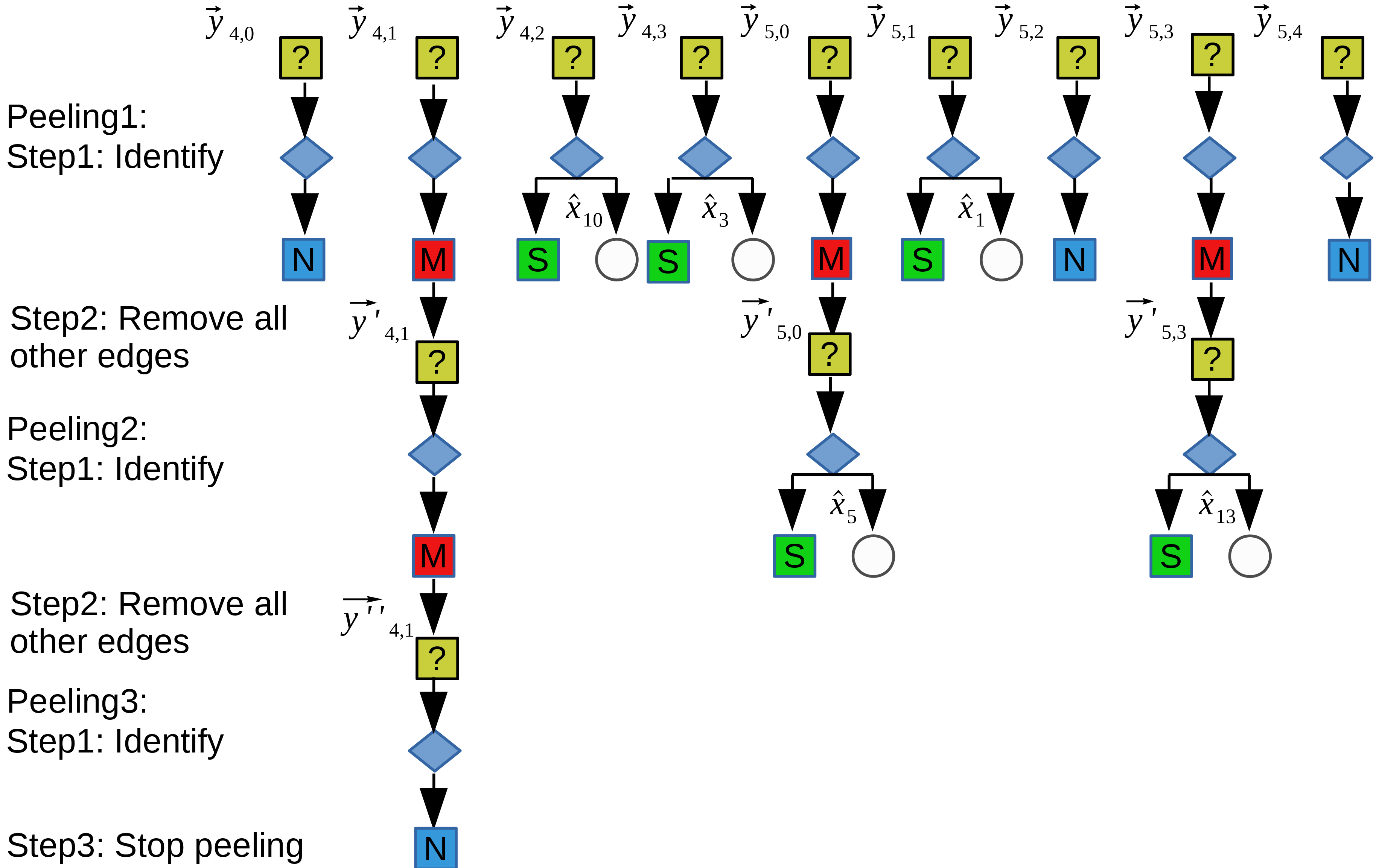}
\caption{The example of the process of the genie-assisted-peeling decoder.	\label{fig4}}
\end{figure} 

Solution with method 2: The sparse-graph decoder by packet erasure channel method is more efficient. From Figure \ref{fig4}, we see all processes need thirteen(9+3+1) identifications in three peelings, so it is not very efficient. As we can see, spectrum recovery can transform into a problem of decoding over sparse bipartite graphs using the peeling platform. The problem of decoding over sparse bipartite graphs has been well studied in the coding theory literature. From the coding theory literature, we know that several sparse-graph code constructions are low-complexity and capacity-achieving by using the erasure channels method. So we use the erasure channels method to improve efficiency. 

We construct a bipartite graph with $K$ variable nodes on the left and $N_b$ check nodes on the right. Each efficient left node $\hat x_f$ connects exactly $d$ right nodes $\overrightarrow{y}_{B_j,i}$'s in $d$ times' cycles and the set of $N_b$ check nodes are assigned to $d$ subsets with the No.$i$'s subset having $B_i$ check nodes. The example of the "balls-and-bins" model defined as above is shown in Fig 2. In Figure \ref{fig3}, there are $K(K=5)$ variable nodes on the left and $N_b(N_b=9)$ check nodes on the right and each left node $\hat x_f$ connects exactly $d(d=2)$ right nodes, the set of check nodes are assigned to two subsets($\{\overrightarrow{y}_{4,0},\overrightarrow{y}_{4,1},\overrightarrow{y}_{4,2},\overrightarrow{y}_{4,3}\},\{\overrightarrow{y}_{5,0},\overrightarrow{y}_{5,1},\overrightarrow{y}_{5,2},\overrightarrow{y}_{5,3},\overrightarrow{y}_{5,4}\}$). If one left node is selected, $d$ number of its neighboring right nodes will be determined. If one right node is selected, $L$ number of its neighboring left nodes will be determined as well. These two corresponding nodes are connected through an edge(In the graph, if the variable nodes on the left is an efficient node, the edge is a solid line. Otherwise, the edge is a dotted line or be omitted). A directed edge $\overrightarrow{e}$ in the graph is represented as an ordered pair of nodes such as $\overrightarrow{e}=\{  V,C  \}$ or $\overrightarrow{e}=\{  C,V   \}$, where $V$ is a variable node and $C$ is a check node. A path in the graph is a directed sequence of directed edges such as $\left\{ \overrightarrow{e_1}, \dots, \overrightarrow{e_l} \right\}$, where the end node of the previous edge of the path is the start node of the next edge of the path. Depth $l$ is defined as the length of the path also means the number of directed edges in the path. The induced subgraph $N_{V}^{l}$ is the directed neighborhood of depth $l$ of left node $V$. It contains all the edges and nodes on paths $\overrightarrow{e_1}, \dots, \overrightarrow{e_l}$ starting at node $V$. It is a tree-like graph which can be seen in Figure \ref{fig5}. In Figure \ref{fig5}, it shows subgraph $N_{V}^{l}$ starting at $\hat x_2$ with depth $l$ is equal to two, three, four, five, six, respectively and subgraph $N_{V}^{l}$ starting at $\hat x_7$ with depth $l$ is equal to six. Under these definitions, we get the steps of packet erasure channel method as follows:

Step 1: Take $O(K)$ random left nodes as start-points, and draw $O(K)$ trees $N_{V}^{1}$ from these start-points. Such as in Figure \ref{fig5}, we choose two left nodes $\hat x_2$ and $\hat x_7$ as start-points. The endpoints of $N_{V}^{1}$ are check nodes, then identify the characteristics of these check nodes. If the check node is a 'zero-ton' bucket, such as $\overrightarrow{y}_{5,2}$, this path stops extending. If the check node is a 'multi-ton' bucket, continue waiting until its left neighboring node identified by other paths. If the check node is a 'single-ton' bucket, such as $\overrightarrow{y}_{4,2}$ and $\overrightarrow{y}_{4,3}$, continue to connect its only one efficient left neighboring node. Then get the tree $N_{V}^{2}$ through expending these left nodes. Their $(d-1)$ number of new right neighboring nodes will be determined through these left nodes, then get the tree $N_{V}^{3}$ by expending these right nodes.

$\dots$

Step $p$: For each start-point $V$, we have got tree $N_{V}^{2p-1}$ from the last step. The endpoints of $N_{V}^{2p-1}$ are check nodes. We should remove the contributions of some frequencies gained by the previous paths at first, then identify the characteristics of these check nodes modified. For example, before identifying $\overrightarrow{y}_{5,0}$, the endpoint of $N_{\hat{x}_2}^{3}$, we should remove the contributions of $\hat x_{10}$. Furthermore, before identifying $\overrightarrow{y}_{4,1}$, the endpoint of $N_{\hat{x}_2}^{5}$, we should remove the contributions of $\hat x_{13}$ and $\hat x_{5}$. If the check node modified is a 'zero-ton' bucket, this path stop extending. If the check node modified is a 'multi-ton' bucket, continue waiting until its left neighboring node identified by other paths. If the check node modified is a 'single-ton' bucket, continue to connect its only one efficient left neighboring node. Then get the tree $N_{V}^{2p}$ through expending these left nodes. Their $(d-1)$ number of new right neighboring nodes will be determined through these left nodes, then get the tree $N_{V}^{2p+1}$ by expending these right nodes.

If the number of left nodes identified equal to $K$, it means the spectrum recovery has been successful. The example is shown in Figure \ref{fig5}. In Figure \ref{fig5}, from the beginning of start-point $\hat x_2$ and $\hat x_7$, we can get two trees $N_{\hat{x}_2}^{6}$ and $N_{\hat{x}_7}^{6}$(depth=6) through three expanding by three steps. From the graph, we can obtain all five variable nodes. All processes need six(4+4-2) identifications, far less than thirteen identifications of method 1.
\begin{figure}[H]
\centering
\includegraphics[width=8.5 cm]{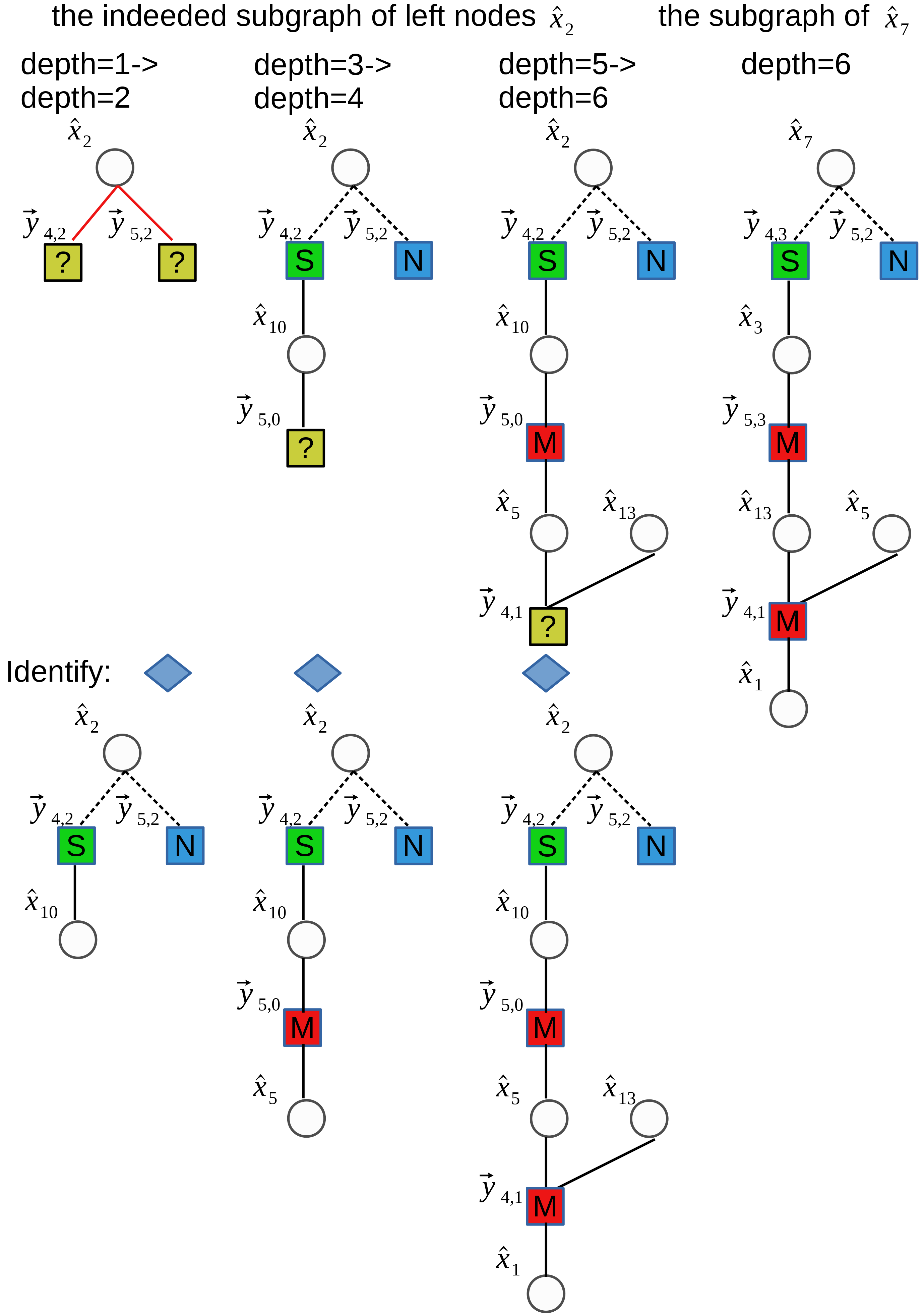}
\caption{The example of the process of the sparse-graph decoder by packet erasure channel method.	\label{fig5}}
\end{figure} 

\subsection{Iterative framework based on the binary tree search method}
Thirdly, we introduce the iterative framework based on the binary tree search method. The example is shown in Figure \ref{fig6}. 
\begin{figure}[H]
\centering
\includegraphics[width=15 cm]{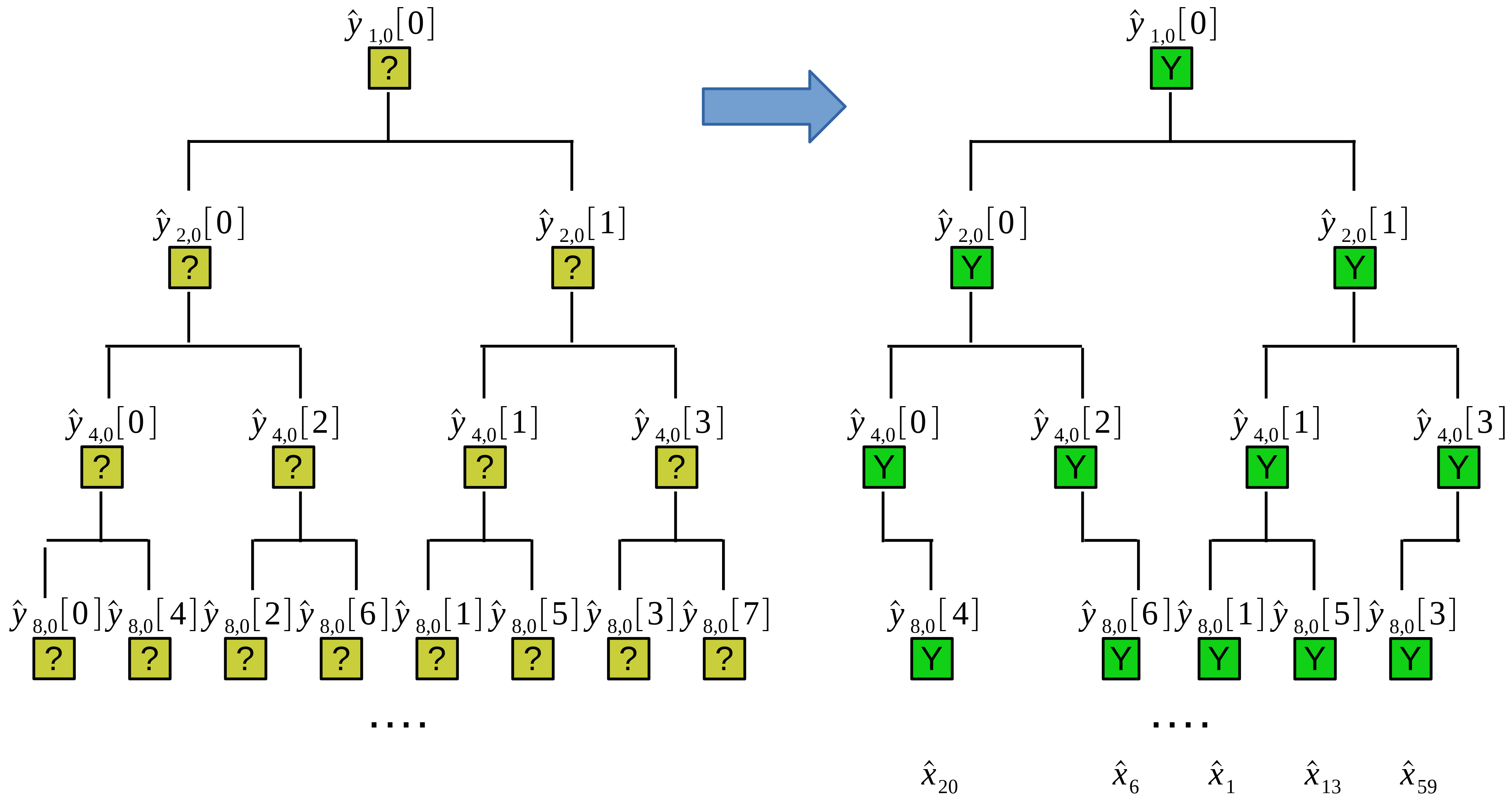}
\caption{The example of the iterative framework based on the binary tree search method.	\label{fig6}}
\end{figure} 
Consider a sixty-four($N = 64$) size signal $\hat{x}$ that has only five($K = 5$) significant coefficients $\hat x _1,\hat x _6,\hat x _{13},\hat x _{20},\hat x _{59}>>0$, while the rest of the coefficients are approximately equal to zero. The node of the first layer of the binary tree is $\hat y _{1,0 }[0]=\sum_{j=0}^{N-1}{\hat {x} _{j} }$, which aliasing $N$ frequencies. The nodes of the second layer of the binary tree are $\hat y _{2,0 }[0]=\sum_{j=0}^{N/2-1}{\hat {x} _{2j} }$ and $\hat y _{2,0 }[1]=\sum_{j=0}^{N/2-1}{\hat {x} _{2j+1} }$, which aliasing $N/2$ frequencies. The nodes of the third layer of the binary tree are $\hat y _{4,0 }[0],\hat y _{4,0 }[1],\hat y _{4,0 }[2],\hat y _{4,0 }[3]$, which aliasing $N/4$ frequencies. The nodes of the No.$(d+1)$'s layer of the binary tree are $\hat y _{2^d,0 }[0],\hat y _{2^d,0 }[1]\dots\hat y _{2^d,0 }[2^d-1]$, which aliasing $N/(2^d)$ frequencies.

The insight of binary tree search is as follows: 1) The frequency of aliasing is gradually dispersed with the expansion of binary tree layers. The number of efficient nodes of No.$(d+1)$'s layer is defined as $m_d$ and the final number will be approximately equal to the sparse number $K$ with the expansion. In Figure \ref{fig6}, $m_0=1, m_1=2, m_2=4, m_3=5$, and the final number of efficient nodes of No.4 's layer $m_3$ is equal to the sparse number $K$. 2) If the parent node exists, there may be one or two child nodes. If the parent node does not exist, there is no need to continue the binary search for this node. For the example of Figure \ref{fig6}, parent node $\hat y _{1,0 }[0]$ exists, so at least one of the two child nodes $\hat y _{2,0 }[0]$ and $\hat y _{2,0 }[1]$ exists. By the same reason, parent node $\hat y _{2,0 }[0]$ and $\hat y _{2,0 }[1]$ exist, so at least one of the two child nodes  $\hat y _{4,0 }[0]$ and $\hat y _{4,0 }[2]$ exists, and at least one of the two child nodes $\hat y _{4,0 }[1]$ and $\hat y _{4,0 }[3]$ exists. On the contrary, parent node $\hat y _{8,0 }[0]$ does not exist, so its two child nodes $\hat y _{16,0 }[0]$ and $\hat y _{16,0 }[8]$ does not exist either. Inspired by these two ideas, the steps of of binary tree search are as follows:

Step 1: Calculate the node $\hat y _{1,0 }[0]$ of the first layer, get $m_0$ and efficient nodes' distribution in this layer.

Step 2: According to the last result, calculate the node $\hat y _{2,0 }[0]$ and $\hat y _{2,0 }[1]$ of the second layer selectively, get $m_1$ and efficient nodes' distribution in this layer.

$\dots$

Step $d+1$: According to the last result, calculate the node $\hat y _{2^d,0 }[0],\hat y _{2^d,0 }[1]$, $\dots$, $\hat y _{2^d,0 }[2^d-1]$ of the No.$(d+1)$'s layer selectively, get $m_d$ and efficient nodes' distribution in this layer. 

We don't need to start from step 1, we can start from step $p(p=O(\text{log}K))$. With the binary tree search, if $m_d$ gained by step $d+1$ is approximately equal to the sparse number $K$, the binary tree search is finished. In Figure \ref{fig6}, $m_3=5=K$, the binary tree search is finished after the fourth layer by step 4. According to the efficient nodes' distribution of No.$(d+1)$'s layer, we can solve the frequency recovery problem of each 'single-ton' bucket. Finally, the $K$ frequency recovery problem is solved by combining the efficient frequency in each bucket. In Figure \ref{fig6}, $\hat y _{8,0 }[4],\hat y _{8,0 }[6],\hat y _{8,0 }[1],\hat y _{8,0 }[5] , \hat y _{8,0 }[3]$ exists,
it represents that each of the five buckets has exactly one effective frequency. Through the frequency recovery of each 'single-ton' bucket, we obtain  $\hat x _1,\hat x _6,\hat x _{13},\hat x _{20},\hat x _{59}$ individually in every bucket. Finally, we obtain $\hat{x}$ of five elements. This algorithm involves three subproblems as follows:

Subproblem 1: How to calculate $\hat y _{2^d,0 }[0],\hat y _{2^d,0 }[1]$, $\dots$, $\hat y _{2^d,0 }[2^d-1]$ selectively according to the last result?

Solution: The runtime of calculating all $2^d$ number of $\hat y _{2^d,0 }[i]$ is $2^d\text{log}2^d=d2^d$ by using the FFT algorithm. The number of effective nodes in the upper layer is $m_{d-1}$, so the maximum number of effective nodes in this layer is $2m_{d-1}$. The formula for calculating each effective node is Equation (\ref{Eq23}), so the runtime of calculating these $2m_{d-1}$ nodes is $2m_{d-1}2^d$. Therefore, when $2m_{d-1} <=d$, use the Equation (\ref{Eq23}) to calculate $2m_{d-1}$ nodes. Otherwise, use the FFT algorithm to calculate all nodes.
\begin{equation}
\hat{y}_{2^d,0}[i]=\frac{1}{2^d}\sum_{j=0}^{2^d-1}(D_LS_0x)_j \omega_{2^d}^{ij}	 \label{Eq23}
\end{equation}

Subproblem 2: The condition of stopping the binary tree search.

Solution: If the stop condition is defined as $m_{d} = K$, in the worst case, the condition can not be satisfied until $d$ is very large. For example if $\hat{x}_{0}$ and $\hat{x}_{N/2}$ are significant frequencies, after the decomposition of the first layer's node $\hat y _{1,0 }[0]$, the second layer's node $\hat y _{2,0 }[0]$, the third layer's node $\hat y _{4,0 }[0]$,$\dots$, their child node $\hat y _{2,0 }[0],\hat y _{4,0 }[0],\hat y _{8,0 }[0],\dots$ are still aliasing until $d = \text{log}N$. Therefore, threshold $= 0.75K$ can be considered, which means that the $K$ frequencies are put into $0.75K$ buckets. The aliasing problem can be solved by using the SVD decomposition described in the last paragraph. Its extra sampling and calculating cost may be more less than that of continuous expansion search if maybe no more frequencies are separated by the next layer.

Subproblem 3: The frequency recovery problem of approximately 'single-ton' bucket.

Solution: In the exactly sparse case, the problem can be solved by using the phase encoding method described in the last paragraph. In the general sparse case, the problem can be solved by the CS method or the MMSE estimator described in the last paragraph.

The binary tree search method is especially suitable for small $K$ or small support. In the case of small support, the frequencies must be scattered into every different bucket. For example, only eight consecutive frequencies of the 1028 frequencies $\hat{x}_{0},\dots,\hat{x}_{1027}$ are significant frequencies, and their eight locations are equal to 0mod8, 1mod8, 2mod8, $\dots$, 7mod8, respectively. In this way, $m_3$ = 8 can be obtained in the fourth layer, and the binary tree search can be stopped by step 4. This method also has the advantage that it is a deterministic algorithm. Finally, the distribution of at most one frequency in one bucket must happen with the layer's expansion, unlike some probabilistic algorithms such as sFFT-DT algorithms.
\section{Algorithms analysis	in theory	\label{Sect4}}
As mentioned above, the goal of frequency bucketization is to decrease runtime and sampling complexity in the advantage of low dimensions. After bucketization, the filtered spectrum $\hat{y}_{L,\tau,\sigma}$ can be obtained by the original signal $x$. Then through three frameworks introduced above, it is sufficient to recover the spectrum $\hat{x}$ of the filtered spectrum $\hat{y}_{L,\tau,\sigma}$ in their own way. In this section, we introduce and analyze the performance of corresponding algorithms including sFFT-DT1.0 algorithm, sFFT-DT2.0 algorithm, sFFT-DT3.0 algorithm, FFAST algorithm, R-FFAST algorithm and DSFFT algorithm.
\subsection{The sFFT-DT algorithms by the one-shot framework}
The block diagram of the sFFT algorithms by the one-shot framework is shown in Figure \ref{fig1}. We explain the details and analyze the performance in theory as below:

Stage 1 Bucketization: Run $(3a_m+3P)$ times' round for set $\tau =\{ \tau_0=-a_m,\dots,\tau_{3a_{m}-1}=2a_{m}-1\}$ and set $\tau =\left\{ \tau_0=r_1,\tau_1=r_2,\dots, \tau_{P-1}=r_{P}          \right\}$ , Calculate $\hat{y}_{L,\tau}=\mathbf{F}_{B}\mathbf{D}_{L}\mathbf{S}_{\tau}x$. It costs $(3a_m+3P)B\text{log}B$ runtime and $(3a_m+3P)B$ samples.

Stage 2-Step 1 Obtain the number of significant frequencies in every bucket: By computing the  SVD for the matrix $\mathbf{M}_{a_m}$ in every bucket, there are $a_m$ singular values for each bucket. Collect all $B·a_m$ singular values from all buckets and index each singular value. The first $K$ largest singular values will vote the number of significant frequencies in every bucket. It costs $a_m^3B$ runtime.

Stage 2-Step 2 Locate the position by method 1: In the exactly sparse case, in a bucket waiting to be solved, firstly obtain $C \approx (\mathbf{M}_a^{-1}) M_s$, then obtain $a$ number of roots of $P(z)=0$. The sFFT-DT1.0 algorithm uses this method, and it costs $a_m^2K$ runtime in this step under the assumption of $a_m\leq 4$.

Stage 2-Step 2 Locate the position by method 2: In a bucket waiting to be solved, firstly obtain $C \approx (\mathbf{M}_a^{-1}) M_s$, try every possible $z_j=\omega^{f_j}$ to the Equation (\ref{Eq12}), the first $a$ number of smallest $|P(z)|$ is the solution. The sFFT-DT2.0 algorithm uses this method, and it costs $a_mKL$ runtime in this step.

Stage 2-Step 2 Locate the position by method 3: In a bucket waiting to be solved, firstly compute the SVD of the matrix $\mathbf{Y}$: $\mathbf Y = \widetilde{\mathbf V} \Sigma \mathbf V ^*$, then get the set of nonzero$^2$ ordinary eigenvalues of $(\mathbf{V_1}^+)\mathbf{V_2}$. The sFFT-DT3.0 algorithm uses this method, and it costs $a_m^3K$ runtime in this step.

Stage 2-Step 3 Estimate the value by CS solver: In a bucket waiting to be solved, 1) Through the positions $f_j$'s gained by step 2, get the possible value of $z_j$'s as follows ${z}'_j,{z}''_j,{z}'''_j$, then obtain 3$a$ vectors. The dictionary $\mathbf{D}$ contains $3a$ vectors listed above. 2) From $3a$ atoms of the dictionary matrix $\mathbf{D}$, find $a$ number of atoms that best matches the measurements' residual error. Select these $a$ number of atoms to construct a new sensing matrix $\hat{\mathbf{\Phi}}$. 3) Obtain $\hat S$ by the support of $\hat S= \text{argmin}\left \| y-\hat{\Phi } \hat S \right \|_2$ through the LSM. 4) If the residual error meets the requirement, or the number of iterations reaches the assumption, or the residual error becomes larger, the iteration will be quit, otherwise continue to step 2. It costs $3a_m^2PK$ runtime in this step.
\begin{lemma}
Suppose $a_m=4, B=32K, L= N/32K,  P=3a_m =12$ and $N/K\geq 32000$, it costs $O(K\text{log}K)$ runtime and $O(K)$ samples in sFFT-DT1.0 algorithm. It costs $O(K\text{log}K+N)$ runtime and $O(K)$ samples in sFFT-DT2.0 algorithm. It costs $O(K\text{log}K)$ runtime and $O(K)$ samples in sFFT-DT3.0 algorithm.
\end{lemma}
\begin{proof}[Proof of Theorem 4]
In sFFT-DT1.0 algorithm, it totally costs $(3a_m+3P)B\text{log}B+a_m^2K+3a_m^2PK=O(K\text{log}K)$ runtime and $(3a_m+3P)B=O(K)$ samples. In sFFT-DT2.0 algorithm, it totally costs $(3a_m+3P)B\text{log}B+a_mKL+3a_m^2PK=O(K\text{log}K+N)$ runtime and $(3a_m+3P)B=O(K)$ samples. In sFFT-DT3.0 algorithm, it totally costs $(3a_m+3P)B\text{log}B+a_m^3K+3a_m^2PK=O(K\text{log}K)$ runtime and $(3a_m+3P)B=O(K)$ samples. Notes: the Theorem 4 has an assumption that $P\geq a_m\text{log}_{10}L$ because of the necessary number of measurements for CS recovery. It means if $N/32K \geq 10^3$, $B$ cannot continue to maintain to equal to $32K$, it has to increase to ensure that $L=N/B\leq 10^3$
\end{proof}

\subsection{The FFAST algorithms by the peeling framework}
The block diagram of the sFFT algorithms by the peeling framework is shown in Figure \ref{fig2}. We explain the details and analyze the performance in theory as below:

Stage 1 Encoding by Bucketization: With the assumption of $K<N^{1/3}$, we run three bucketization cycles, and each cycle needs $R$ times' round. The number of buckets is equal to $B_1, B_2, B_3$ individually in three cycles, correspondingly the size of one bucket is equal to $L_1, L_2, L_3$ individually in three cycles. In the exactly sparse case, in the first peeling, we use the FFAST algorithm with $R=3$ and the set $\tau =\left\{ \tau_0=0,\tau_1=1,\tau_2=2 \right\}$ in three rounds. In the general sparse case, in the first peeling, we use the R-FFAST algorithm with $R = Cm= O(\text{log}^{4/3}N)$ and the set $\tau =\left\{ r_0, r_0+1,\dots,r_1, r_1+2,\dots,r_{C-1}, r_{C-1}+2^{C-1},\dots, r_{C-1} + 2^{C-1}(m-1)  \right\}$ in $R$ rounds. After bucketization, we get the vectors $\overrightarrow{y}_{B_1,0},\dots,\overrightarrow{y}_{B_1,B_1-1},\overrightarrow{y}_{B_2,0},\dots$, $\overrightarrow{y}_{B_2,B_2-1},\overrightarrow{y}_{B_3,0},\dots,\overrightarrow{y}_{B_3,B_3-1}$. It costs $3(B_1 \text{log} B_1+B_2 \text{log} B_2+B_3 \text{log} B_3)$ runtime and $3(B_1+B_2+B_3)$samples in stage 1 for the FFAST algorithm. It costs $(\text{log}^{4/3}N)(B_1\text{log}B_1+B_2\text{log}B_2+B_3\text{log}B_3)$ runtime and $(\text{log}^{4/3}N)(B_1+B_2+B_3)$ samples in stage 1 for the R-FFAST algorithm.

Stage 2 Decoding by several peeling: The subsequent calculation complexity is mainly composed of two parts: the judgment of the old bucket and the generation of the new bucket, the runtime of the generation of the new bucket can be ignored, so the computational complexity is mainly determined by the runtime of bucket judgment. Since the main buckets are identified as 'zero-ton' bucket or 'single-ton' bucket during the first peeling, the number of buckets that need to be calculated in the subsequent peeling iteration is not large, so the runtime is mainly determined at the first peeling. For the FFAST algorithm, there are $(B_1+B_2+B_3)$ buckets to be judged in the first peeling, and the runtime of each bucket is $O(4)$ proved in Sec. \ref{Sect3}. For the R-FFAST algorithm, the runtime of each bucket is $O(4\text{log}^{4/3}N)$ proved in Sec. \ref{Sect3}.
\begin{lemma}
With the assumption of $K<N^{1/3}$, Suppose $B_1 = O(K), B_2 = O(K), B_3 = O(K)$, it costs $O(K\text{log}K)$ runtime and $O(K)$ samples in the FFAST algorithm. It costs $O(K\text{log}^{7/3}N)$ runtime and $O(K\text{log}^{4/3}N)$ samples in the R-FFAST algorithm.
\end{lemma}
\begin{proof}[Proof of Theorem 5]
For the FFAST algorithm, it totally costs $3(B_1\text{log}B_1+B_2\text{log}B_2+B_3\text{log}B_3)+ 4(B_1+B_2+B_3) = O(K\text{log}K)$ runtime and $3(B_1+B_2+B_3)=O(K)$ samples. For the R-FFAST algorithm, it totally costs $(\text{log}^{4/3}N)(B_1\text{log}B_1+B_2\text{log}B_2+B_3\text{log}B_3) + 4(\text{log}^{4/3}N)(B_1+B_2+B_3) = O(K\text{log}^{7/3}N)$ runtime and $(\text{log}^{4/3}N)(B_1+B_2+B_3)=O(K\text{log}^{4/3}N)$  samples. Notes: From the paper \cite{IEEEexample:Pawar2018}, with the assumption of $K<N^{1/3}$, if the number of buckets $B_1, B_2, B_3$ are co-prime number and approximately equal to $O(K)$, chain number $R=O(\text{log}^{4/3}N)$, and in every 'single-ton' estimator, the number of iterations $C = O(\text{log} N)$, the number of rounds per iteration $m = O(\text{log}^{1/3}N)$, then the algorithm can correctly recover the spectrum with a high probability(probability of successful $\geq (1-1/K))$.
\end{proof}

\subsection{The DSFFT algorithm by the binary tree search framework}
The first stage of DSFFT algorithm is finding exactly $O(K)$ efficient buckets through the binary tree search method. When $m_d$ gained by step $(d+1)$ is equal to the sparse number $K$, the binary tree search is finished. Now we start to calculate the probability in this case. The total number of buckets is $2^d=B$, the number of aliasing frequencies in each bucket is $L=N/B=N/2^d$. In other words, the right case is $K$ large frequencies of totally $N$ frequencies allocated to $B$ buckets, and there is no aliasing in the $B$ distributed buckets. The probability $P$ is calculated as follows: Firstly, analyze the first bucket, the probability $P_1$ is the case of the non-aliasing case divided by the total case, $P_1=\frac{(N-K)\cdot(N-K-1)\dots(N-K-L+2)}{(N-1)\cdot(N-2)\dots(N-L+1)}$. Secondly, analyze the second bucket, the probability $P_2=\frac{(N-K-L+1)\cdot(N-K-L)\dots(N-K-2L+3)}{(N-L-1)\cdot(N-L-2)\dots(N-2L+1)},\dots$. At last, analyze the No.$K$'s bucket, $P_K=\frac{(N-(K-1)L-1)\dots (N-KL+1)}{(N-(K-1)L-1)  \dots (N-KL+1)}$. So the complete probability $P=P_1 P_2 \dots P_K=\frac{(N-L)(N-2L)\dots(N-(K-1)L)}{(N-1)(N-2)\dots(N-K+1)}=\frac{L^{(K-1)}((B-1)!)((N-K)!)}{((N-1)!)((B-K)!)}$. For example, consider $N=1000, K=5, B=10, L=100$, the complete probability $P=(\frac{995\times \dots 897 }{999\times \dots 901 })\dots(\frac{599\times \dots 501 }{599\times \dots 501 })=\frac{995\cdot 990\cdot 985\cdot 980}{999\cdot 998\cdot 997\cdot 996}$. Table \ref{tab:1} shows the probability $P_{B_0},P_{B_1},P_{B_2},P_{B_3}$ of different $K$ in the case of $B_0=K,B_1=2K,B_2=4K,B_3=8K$. When $B_0=K, P_{B_0}=(\frac{N}{K})^{(K-1)}  \frac{((K-1)!)((N-K)!)}{(N-1)!}$. When $B_1=2K, P_{B_1}=(\frac{N}{2K})^{(K-1)}  \frac{((2K-1)!)((N-K)!)}{((N-1)!)((K)!)}$. When $B_2=4K,P_{B_2}=(\frac{N}{4K})^{(K-1)}  \frac{((4K-1)!)((N-K)!)}{((N-1)!)((3K)!)} \dots$.

\begin{table}[H] 
\caption{The probability $P_{B_0},P_{B_1},P_{B_2},P_{B_3}$ of different $K$ in the case of $B$ from $K$ to $8K$}
\label{tab:1}
\setlength{\tabcolsep}{6pt}
\begin{tabular}{|p{40pt}|p{25pt}|p{60pt}|p{70pt}|p{70pt}|p{70pt}|}
\hline
probability& 
$K=1$& 
$K=2$& 
$K=4$& 
$K=6$& 
$K=8$  \\
\hline
$P_{B_0}$&
1&
$\approx \frac{1}{2}=0.5$&
$\approx \frac{3!}{4^3}\approx 0.0938$&
$\approx \frac{5!}{6^5}\approx 0.0154$&
$\approx \frac{7!}{8^7 }\approx 0.0024$ \\
$P_{B_1}$&
1&
$\approx \frac{3}{4}=0.75$&
$\approx \frac{7!}{8^3 4!}\approx 0.4101$&
$\approx \frac{11!}{12^5 6!}\approx 0.2228$&
$\approx \frac{15!}{16^7 8!}\approx 0.1208$ \\
$P_{B_2}$&
1&
$\approx \frac{7}{8}=0.875$&
$\approx \frac{15!}{16^3 12!}\approx 0.6665$&
$\approx \frac{23!}{24^5 18!}\approx 0.5071$&
$\approx \frac{31!}{32^7 24!}\approx 0.3857$ \\
$P_{B_3}$&
1&
$\approx \frac{15}{16}=0.9375$&
$\approx \frac{31!}{32^3 28!}\approx 0.8231$&
$\approx \frac{47!}{48^5 42!}\approx 0.7224$&
$\approx \frac{63!}{64^7 56!}\approx 0.6340$ \\

\hline

\end{tabular}
\label{tab1}
\end{table}

As we can see from Table \ref{tab:1}, the DSFFT algorithm has high efficiency only when $K$ is very small. The total complexity can be calculated according to the sum of probability times conditional complexity. So the total runtime complexity is $P_{B_0} K \text{log}K +(P_{B_1}-P_{B_0}) 2K\text{log}2K +(P_{B_2}-P_{B_1}) 4K\text{log}4K   +\dots = \sum_{j=0}^{j=\text{log}(N/K)}(P_{B_j}-P_{B_{j-1}})2^jK \text{log}(2^jK)$ for $P_{B_{-1}}=0$, and the total sampling complexity is $P_{B_0} K+(P_{B_1}-P_{B_0}) 2K +(P_{B_2}-P_{B_1}) 4K  +\dots = \sum_{j=0}^{j=\text{log}(N/K)}(P_{B_j}-P_{B_{j-1}})2^jK$ for $P_{B_{-1}}=0$.

\subsection{Summary of six algorithms in theory}
After analyzing three types of frameworks, six corresponding algorithms, Table \ref{tab:2}\footnote{The performance of algorithms using the aliasing filter is got as above. The performance of other algorithms is got from paper[7]. The analysis of robustness will be explained in the next section.} can be concluded with the additionl information of other sFFT algorithms and fftw algorithm.

\begin{table}[H] 
\caption{The performance of fftw algorithm and sFFT algorithms in theory}
\label{tab:2}
\setlength{\tabcolsep}{4pt}
\begin{tabular}{|p{50pt}|p{150pt}|p{120pt}|p{40pt}|}
\hline
algorithm& 
runtime complexity& 
sampling complexity & robustness \\
\hline
sFFT1.0&
$O(K^{\frac{1}{2}}N^{\frac{1}{2}}\ \text{log} ^{\frac{3}{2}}N)$& 
$N\left(1-\left(\frac{N-w}{N}\right)^{\ \text{log} N}\right)$& 
medium \\

sFFT2.0&
$O(K^{\frac{2}{3}}N^{\frac{1}{3}} \ \text{log} ^{\frac{4}{3}} N)$& 
$N\left(1-\left(\frac{N-w}{N}\right)^{\ \text{log} N}\right)$& 
medium  \\

sFFT3.0&
$O(K \ \text{log} N) $& 
$O(K \ \text{log} N) $& 
none  \\

sFFT4.0&
$O(K \ \text{log} N \ \text{log} _8 (N/K) )$& 
$O(K \log N \ \text{log} _8 (N/K) )$& 
bad  \\

MPFFT&
$O(K \log N \ \text{log} _2 (N/K) )$& 
$O(K \log N \ \text{log} _2 (N/K) )$& 
good  \\

sFFT-DT1.0&
$O(K \ \text{log} K )$& 
$O(K)$& 
none  \\

sFFT-DT2.0&
$O(K \ \text{log} K + N)$& 
$O(K )$& 
medium  \\

sFFT-DT3.0&
$O(K \ \text{log} K )$& 
$O(K)$& 
medium  \\

FFAST&
$O( K \ \text{log} K) $& 
$O( K ) $& 
none  \\

R-FFAST&
$ O (K \ \text{log} ^{7/3} N)$& 
$O( K \ \text{log} ^{4/3} K) $& 
good  \\

DSFFT&
$\sum_{j=0}^{j=\text{log}(N/K)}(P_{B_j}-P_{B_{j-1}})2^jK \text{log}(2^jK)$& 
$\sum_{j=0}^{j=\text{log}(N/K)}(P_{B_j}-P_{B_{j-1}})2^jK$& 
medium  \\

AAFFT&
$O( K \text{poly} ( \ \text{log} N))$& 
$O( K \text{poly} ( \ \text{log} N))$& 
medium   \\

fftw&
$ O( N \ \text{log} N)$& 
$N$& 
good  \\

\hline

\end{tabular}
\label{tab1}
\end{table}

From Table \ref{tab:2}, we can see the sFFT-DT1.0 algorithm and the FFAST algorithm are the lowest runtime and sampling complexity, but they are non-robustness. Other algorithms using the aliasing filter are good sampling complexity but compare them with other sFFT algorithms it is no advantage in the runtime complexity except sFFT-DT3.0 algorithm. In addition, the use of aliasing algorithm is limited in some aspects. For the algorithms of the one-shot platform, the performance is not very efficient if $N/K\leq 32000$ because of needing the necessary number of measurements for CS recovery. For the algorithms of the peeling platform, the assumption of $K<N^{1/3}$ is required because if there are four or more cycles, it will be very complicated. For the algorithm of the binary tree search framework, $K$ is required to be very small, because only in this case, there is no need to expand the scale of a large binary tree.

\section{Algorithms analysis	in practice	\label{Sect5}}
In this section, we evaluate the performance of six sFFT algorithms using the aliasing filter: sFFT-DT1.0, sFFT-DT2.0, sFFT-DT3.0, FFAST, R-FFAST and DSFFT algorithm. We firstly compare these algorithms' runtime, percentage of the signal sampled and robustness characteristics with each other. Then we compare some of these algorithms' characteristics with other algorithms: fftw, sFFT1.0, sFFT2.0, sFFT-3.0, sFFT-4.0 and AAFFT algorithm. All experiments are run on a Linux CentOS computer with 4 Intel(R) Core(TM) i5 CPU and 8 GB of RAM.
\subsection{Experimental Setup}
In the experiment, the test signals are gained in a way that $K$ frequencies are randomly selected from $N$ frequencies and assigned a magnitude of 1 and a uniformly random phase. The rest frequencies are set to zero in the exactly sparse case or combined with additive white Gaussian noise in the general sparse case, whose variance varies depending on the SNR required. The parameters of these algorithms are chosen so that they can make a balance between time efficiency and robustness.
\subsection{Comparison experiment about different algorithms using the aliasing filter of themselves}
We plot Figure \ref{fig7}(a), \ref{fig7}(b) representing run time vs. Signal Size and vs. Signal Sparsity for sFFT-DT2.0, sFFT-DT3.0, R-FFAST and DSFFT algorithm in the general sparse case\footnote{The general sparse case means SNR=20db. The exactly case is very similar to the general sparse case except the sFFT-DT1.0 algorithm and the FFAST algorithm. The detail of code, data, report can be provided in https://github.com/zkjiang/-/tree/master/docs/sfft project}. As mentioned above, the runtime is determined by two factors. One is how many buckets to cope with, and another is how much time cost to identify frequencies in efficient buckets. So from Figure \ref{fig7}, we can see 1) The runtime of these four algorithms are approximately linear in the log scale as a function of $N$ and non-linear as a function of $K$. 2) Result of ranking the runtime of four algorithms is sFFT-DT3.0$>$sFFT-DT2.0$>$DSFFT$>$R-FFAST. The reason is the method of R-FFAST algorithm is the most time-consuming, the method of DSFFT algorithm is also not efficient, the method of sFFT-DT2.0 is much better, and the method of sFFT-DT3.0 has the highest time performance. 3) $K$ has a certain limitation. The DSFFT algorithm is very inefficient when $K$ is greater than 50, and the R-FFAST algorithm does not work when $K$ is greater than $N^{1/3}$.

\begin{figure} [H]\centering  
\subfigure[Run time vs. Signal Size.] { 
\includegraphics[width=0.45\columnwidth]{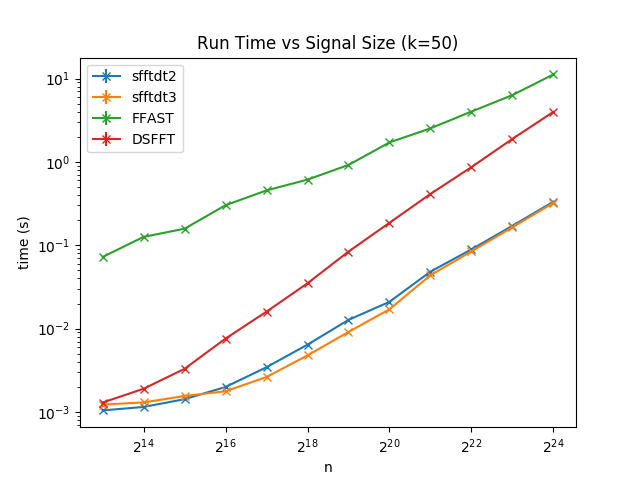}
}   
\subfigure[Run time vs. Signal Sparsity.] { 
\includegraphics[width=0.45\columnwidth]{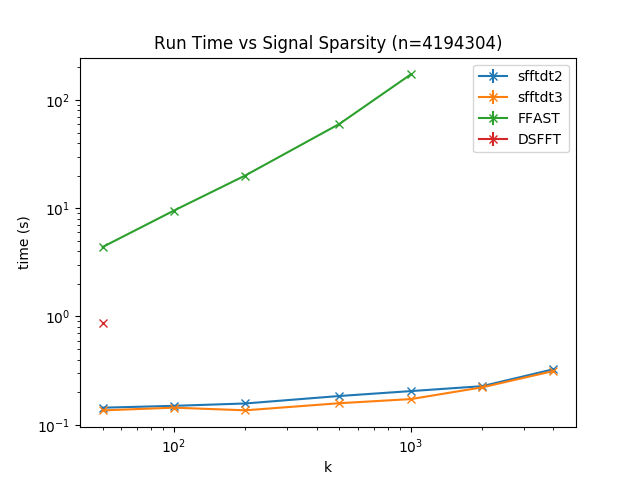}   
}   
\caption{Run time of four algorithms using the aliasing filter in the general sparse case. \label{fig7}}      
\end{figure}

We plot Figure \ref{fig8}(a), \ref{fig8}(b) representing the percentage of the signal sampled vs. Signal Size and vs. Signal Sparsity for sFFT-DT2.0, sFFT-DT3.0 and R-FFAST algorithm in the general sparse case. As mentioned above, the percentage of the signal sampled is also determined by two factors: how many buckets and how many samples sampled in one bucket. So from Figure \ref{fig8}, we can see 1) The percentage of the signal sampled of these three algorithms are approximately linear in the log scale as a function of $N$ and non-linear as a function of $K$. 2) Result of ranking the sampling complexity of three algorithms is R-FFAST$>$sFFT-DT2.0=sFFT-DT3.0 because R-FFAST algorithm uses the principle of CRT, the number of buckets is independent of $K$, and the proportion decreases with the increase of $N$. 3) There is a limit for sFFT-DT2.0 and sFFT-DT3.0 algorithm. When $N/K$ is too large, $B$ can not maintain the linearity of the sparsity $K$. For the R-FFAST algorithm; there is a limit that $K$ cannot be greater than $N^{1/3}$.

\begin{figure} [H]\centering  
\subfigure[Percentage of the Signal Sampled vs. Signal Size.] { 
\includegraphics[width=0.45\columnwidth]{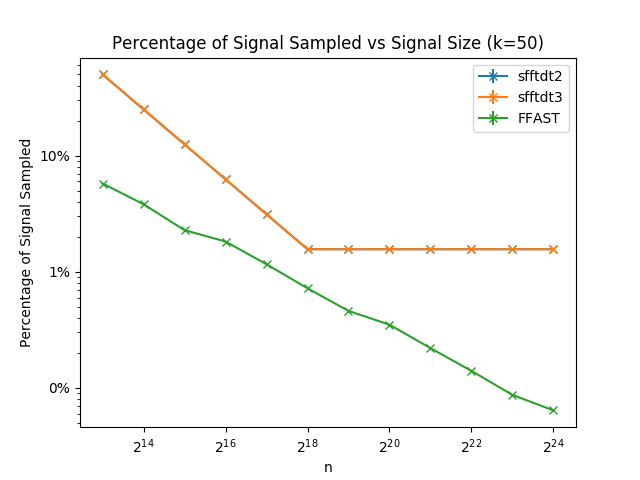}
}   
\subfigure[Percentage of the Signal Sampled vs. Signal Sparsity.] { 
\includegraphics[width=0.45\columnwidth]{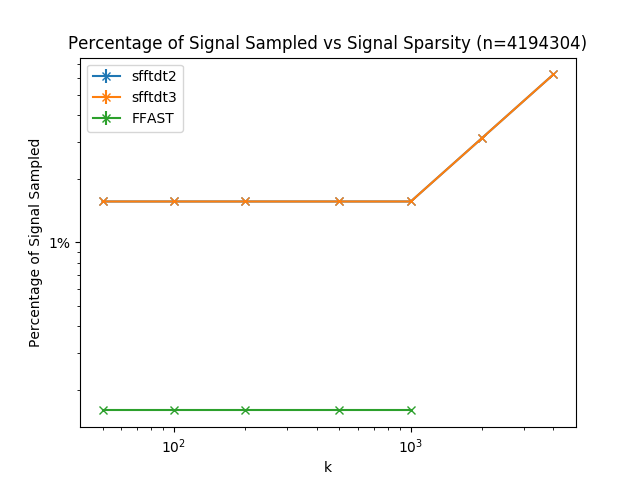}   
}   
\caption{Percentage of the Signal Sampled of algorithms using aliasing filter in the general sparse case. \label{fig8}}    
\end{figure}

We plot Figure \ref{fig9}(a), \ref{fig9}(b) representing run time and $L_1$-error vs SNR for sFFT-DT2.0, sFFT-DT3.0 and R-FFAST algorithm. From Figure \ref{fig9} we can see 1) The runtime of sFFT-DT2.0, sFFT-DT3.0 is approximately equal vs SNR, but the runtime of the R-FFAST algorithm increases with the noise. 2) To a certain extent, these three algorithms are all robust. When SNR is low, only the R-FFAST algorithm satisfies the ensure of robustness. When SNR is medium, the sFFT-DT3.0 algorithm can also meet the ensure of robustness. And only when SNR is bigger than 20db, the sFFT-DT2.0 algorithm can deal with the noise interference. The reason is that the way of binary search is better than other ways in terms of robustness.
\begin{figure} [H]\centering  
\subfigure[Run time vs. SNR.] { 
\includegraphics[width=0.45\columnwidth]{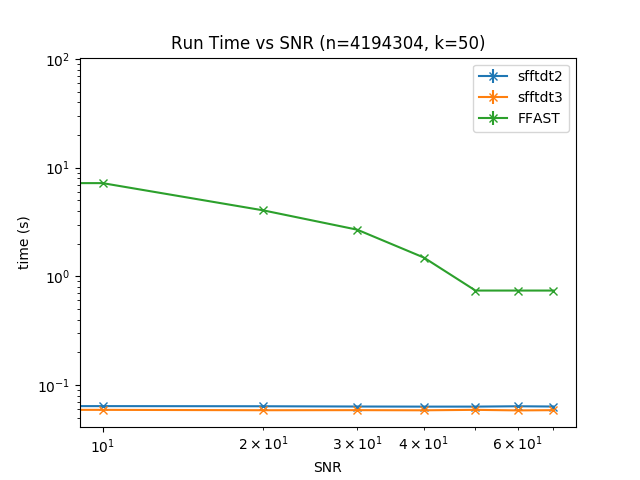}
}   
\subfigure[$L_1$-error vs. SNR.] { 
\includegraphics[width=0.45\columnwidth]{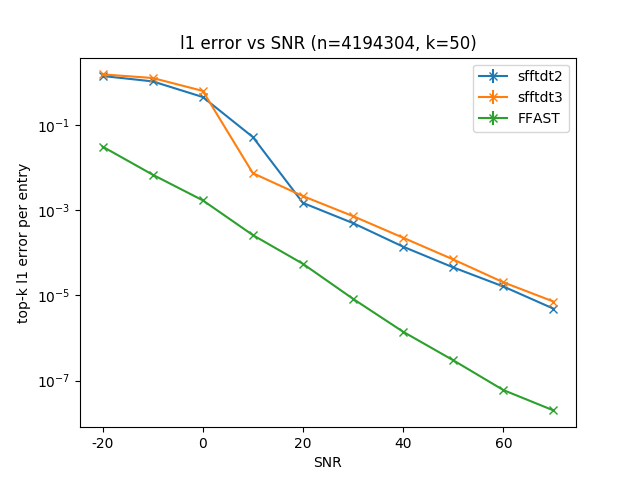}   
}   
\caption{Run time and $L_1$-error of algorithms using aliasing filter in the general sparse case vs. SNR.  \label{fig9}}   
\end{figure}

\subsection{Comparison experiment about algorithms using the aliasing filter and other algorithms}
We plot Figure \ref{fig10}(a) and \ref{fig10}(b) representing run time vs. Signal Size and vs. Signal Sparsity for sFFT-DT3.0, R-FFAST, sFFT2.0, sFFT4.0, AAFFT and fftw algorithm in the general sparse case.\footnote{The sFFT1.0 algorithm is ignored because it is similar to the sFFT2.0 algorithm but not as efficient as it. The sFFT3.0 algorithm is ignored because it is not suitable for general sparsity. The MPSFT algorithm is ignored because it is similar to the sFFT4.0 algorithm but not as efficient as it.} From Figure \ref{fig10}, we can see 1) These algorithms are approximately linear in the log scale as a function of $N$ except fftw algorithm. These algorithms are approximately linear in the standard scale as a function of $K$ except fftw and sFFT-DT3.0 algorithm. 2) Result of ranking the runtime complexity of these six algorithms is sFFT2.0$>$sFFT4.0$>$AAFFT$>$sFFT-DT3.0$>$fftw$>$R-FFAST when $N$ is large. The reason is the Least Absolute Shrinkage and Selection Operator(LASSO) method used in R-FFAST costs a lot of time. Besides, the SVD method and the CS method used in sFFT-DT also cost a lot of time. 2) Result of ranking the runtime complexity of these six algorithms is fftw$>$sFFT-DT3.0$>$sFFT4.0 $>$sFFT2.0$>$AAFFT$>$R-FFAST when $K$ is large. The reason is algorithms using the aliasing filter only need much fewer buckets than algorithms using the flat filter when $K$ is large.
\begin{figure} [H]\centering  
\subfigure[Runtime vs. Signal Size.] { 
\includegraphics[width=0.45\columnwidth]{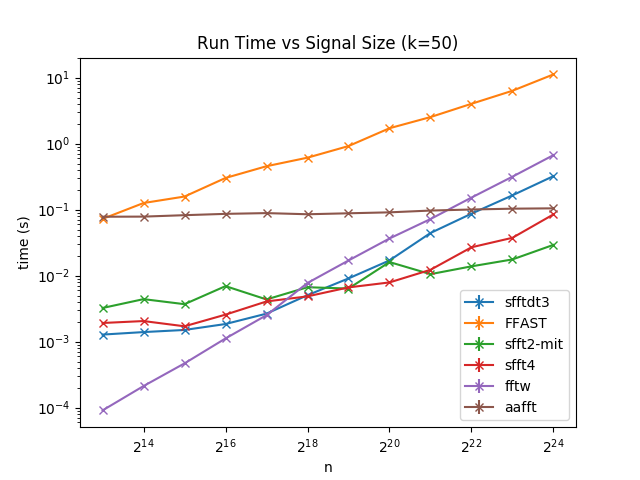}
}   
\subfigure[Runtime vs. Signal Sparsity.] { 
\includegraphics[width=0.45\columnwidth]{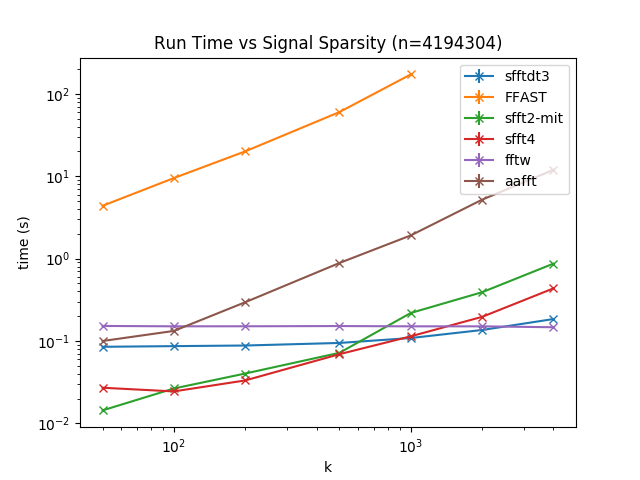}   
}   
\caption{Runtime of two typical algorithms using the aliasing filter and other four algorithms in the general sparse case. \label{fig10}}   
\end{figure}

We plot Figure \ref{fig11}(a) and \ref{fig11}(b) representing the percentage of the signal sampled vs. Signal Size and vs. Signal Sparsity for sFFT-DT3.0, R-FFAST, sFFT2.0, sFFT4.0, AAFFT and fftw algorithm in the general sparse case. From Figure \ref{fig11}, we can see 1) These algorithms are approximately linear in the log scale as a function of $N$ except fftw and sFFT-DT algorithm. The reason is sampling in low-dimension in sFFT algorithms can decrease sampling complexity and it is a limit to the size of the bucket in the sFFT-DT algorithm by using the CS method. These algorithms are approximately linear in the standard scale as a function of $K$ except R-FFAST, sFFT-DT and fftw algorithm. The reason is algorithms using the aliasing filter saving samples by using less number of buckets. 2) Result of ranking the sampling complexity of these six algorithms is R-FFAST$>$sFFT4.0$>$AAFFT$>$sFFT2.0$>$sFFT-DT3.0$>$fftw when $N$ is large. 3) Result of ranking the sampling complexity is sFFT-DT3.0$>$sFFT4.0$>$AAFFT$>$sFFT2.0$>$fftw when $K$ is large.	
\begin{figure} [H]\centering  
\subfigure[Percentage of the signal sampled vs. $N$.] { 
\includegraphics[width=0.45\columnwidth]{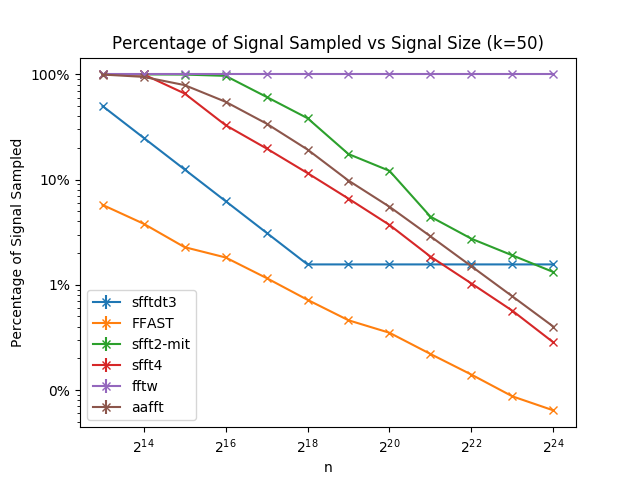}
}   
\subfigure[Percentage of the signal sampled vs. $K$.] { 
\includegraphics[width=0.45\columnwidth]{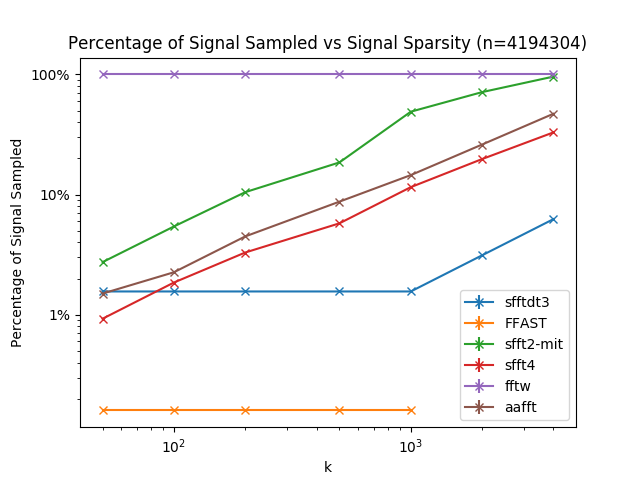}   
}   
\caption{Percentage of the signal sampled of two typical algorithms using the aliasing filter and other four algorithms in the general sparse case.\label{fig11}}   
\end{figure}

We plot Figure \ref{fig12}(a), \ref{fig12}(b) representing run time and $L_1$-error vs SNR for sFFT-DT3.0, R-FFAST, sFFT2.0, sFFT4.0, AAFFT and fftw algorithm. From Figure \ref{fig12}, we can see 1) The runtime is approximately equal vs SNR except the R-FFAST algorithm. 2) To a certain extent, these six algorithms are all robust, but when SNR is low, only fftw and R-FFAST algorithm satisfies the ensure of robustness. When SNR is medium, sFFT2.0, AAFFT and sFFT-DT3.0 algorithm can also meet the ensure of robustness. And only when SNR is bigger than 20db, the sFFT4.0 algorithm can deal with the noise interference. 
\begin{figure} [H]\centering  
\subfigure[Run time vs. SNR.] { 
\includegraphics[width=0.45\columnwidth]{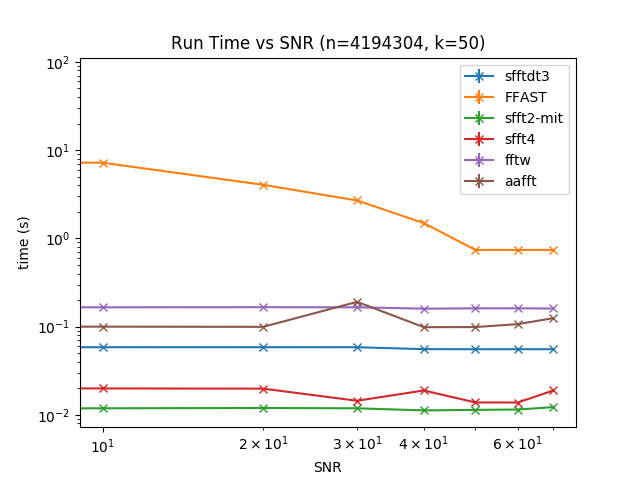}
}   
\subfigure[$L_1$-error vs. SNR.] { 
\includegraphics[width=0.45\columnwidth]{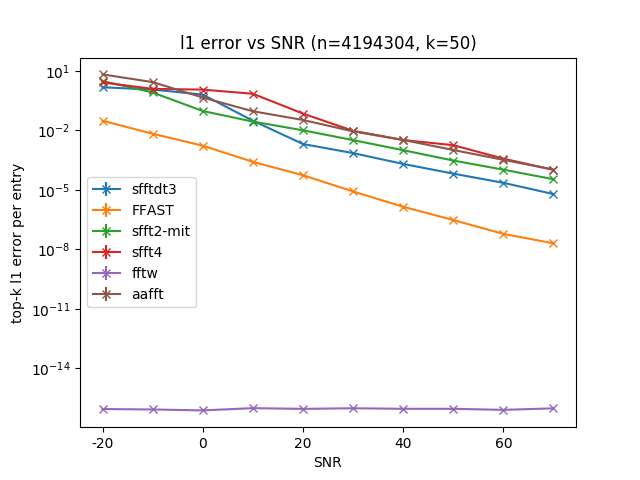}   
}   
\caption{Run time and $L_1$-error of two typical algorithms using the aliasing filter and other four algorithms in the general sparse case vs. SNR.\label{fig12}}   
\end{figure}

\section{Conclusion		\label{Sect6}}
In the first part, the paper introduces the techniques used in the sFFT algorithm including time-shift operation and subsampling operation. In the second part, we analyze six typical algorithms using the aliasing filter in detail. By the one-shot framework based on the CS solver, sFFT-DT1.0 algorithm uses the polynomial method, sFFT-DT2.0 algorithm uses the enumeration method and sFFT-DT3.0 algorithm uses the Matrix Pencil method. By the peeling framework based on the bipartite graph, the FFAST algorithm uses the phase encoding method, the R-FFAST algorithm uses the MMSE estimator method. By the iterative framework, the DSFFT algorithm uses the binary tree search method. We get the conclusion of the performance of these algorithms including runtime complexity, sampling complexity and robustness in theory in Table \ref{tab:2}. In the third part, we make two categories of experiments for computing the signals of different SNR, different $N$ and different $K$ by a standard testing platform and record the run time, percentage of the signal sampled and $L_0,L_1,L_2$ error by using nine different sFFT algorithms in every different situation both in the exactly sparse case and general sparse case. The analysis of the experiment results satisfies theoretical inference.

The main contribution of this paper is 1) Develop a standard testing platform which can test more than ten typical sFFT algorithms in different situations on the basic of an old platform  2) Get a conclusion of the character and performance of the six typical sFFT algorithms using the aliasing filter: sFFT-DT1.0 algorithm, sFFT-DT2.0 algorithm, sFFT-DT3.0 algorithm, FFAST algorithm, R-FFAST algorithm and DSFFT algorithm in theory and practice.


\bibliographystyle{IEEEtran}
\bibliography{IEEEabrv,IEEEexample}

\end{document}